\documentclass[sigconf, nonacm]{acmart}
\usepackage{xcolor}

\usepackage[linesnumbered, ruled, vlined]{algorithm2e}
\usepackage{subfigure}
\usepackage{array}
\usepackage{enumitem}
\setitemize{leftmargin=10pt, noitemsep,topsep=0pt,parsep=0pt,partopsep=0pt}
\everymath{\scriptstyle}
\DeclareMathOperator*{\argmax}{arg\,max}

\newcommand\vldbdoi{XX.XX/XXX.XX}

\newcommand\vldbvolume{14}
\newcommand\vldbissue{1}

\newcommand\vldbavailabilityurl{URL_TO_YOUR_ARTIFACTS}
 
\begin{document}
\title{Efficient Size Constraint Community Search over Heterogeneous Information Networks}

\author{Xinjian Zhang}
\affiliation{%
  \institution{Swinburne University of Technology}
  \city{Melbourne}
  \country{Australia}
}
\email{xinjianzhang@swin.edu.au}

\author{Lu Chen}
\affiliation{%
  \institution{Swinburne University of Technology}
  \city{Melbourne}
  \country{Australia}
}
\email{luchen@swin.edu.au}

\author{Chengfei Liu}
\authornote{Corresponding author}
\affiliation{%
  \institution{Swinburne University of Technology}
  \city{Melbourne}
  \country{Australia}
}
\email{cliu@swin.edu.au}

\author{Rui Zhou}
\affiliation{%
  \institution{Swinburne University of Technology}
  \city{Melbourne}
  \country{Australia}
}
\email{rzhou@swin.edu.au}

\author{Bo Ning}
\affiliation{%
  \institution{Dalian Maritime University}
  \city{Dalian}
  \country{China}
}
\email{ningbo@dlmu.edu.cn}

%
\begin{abstract}
The goal of community search in heterogeneous information networks (HINs) is to identify a set of closely related target nodes that includes a query target node. In practice, a size constraint is often imposed due to limited resources, which has been overlooked by most existing HIN community search works.  
In this paper, we introduce the size-bounded community search problem to HIN data. Specifically, we propose a refined $(k, \mathcal{P})$-truss model to measure community cohesiveness, aiming to identify the most cohesive community of size $s$ that contains the query node. 
We prove that this problem is NP-hard. 
To solve this problem, we develop a novel B\&B framework that efficiently generates target node sets of size $s$. We then tailor novel bounding, branching, total ordering, and candidate reduction optimisations, which enable the framework to efficiently lead to an optimum result. We also design a heuristic algorithm leveraging structural properties of HINs to efficiently obtain a high-quality initial solution, which serves as a global lower bound to further enhance the above optimisations. 
Building upon these, we propose two exact algorithms that enumerate combinations of edges and nodes, respectively.
Extensive experiments on real-world datasets demonstrate the effectiveness and efficiency of the proposed methods.
\end{abstract}

\maketitle

\begingroup
\renewcommand\thefootnote{}\footnote{\noindent This work is licensed under the Creative Commons BY-NC-ND 4.0 International License. Visit \url{https://creativecommons.org/licenses/by-nc-nd/4.0/} to view a copy of this license. For any use beyond those covered by this license, obtain permission by emailing \href{mailto:info@vldb.org}{info@vldb.org}. Copyright is held by the owner/author(s). Publication rights licensed to the VLDB Endowment. \\
\raggedright Proceedings of the VLDB Endowment, Vol. \vldbvolume, No. \vldbissue\ %
ISSN 2150-8097. \\
\href{https://doi.org/\vldbdoi}{doi:\vldbdoi} \\
}\addtocounter{footnote}{-1}\endgroup


\definecolor{orchid}{RGB}{0,0,0}

\newtheoremstyle{exampstyle}
{2pt} 
{2pt} 
{\itshape} 
{} 
{\scshape} 
{.} 
{.2em} 
{} 
\theoremstyle{exampstyle}
\let\lemma\relax
\let\example\relax
\let\definition\relax

\newtheorem{definition}{Definition}
\newtheorem{lemma}{Lemma}
\newtheorem{example}{Example} 
\newtheorem{ruledef}{rule}

\newtheorem{prop}{Property}
\newcommand{\mP}{\mathcal{P}}

\newcommand{\mA}{\mathcal{A}}

\newcommand{\mR}{\mathcal{R}}

\newtheorem{obv}{Observation}
\newtheorem{stopC}{Stop condition}

\SetAlFnt{\small\normalfont}
\SetAlCapHSkip{0em}
\SetAlgoSkip{}
\setlength\floatsep{1.25\baselineskip plus 3pt minus 3pt}

\setlength\textfloatsep{1.25\baselineskip plus 3pt minus 3pt}
\setlength\intextsep{1.25\baselineskip plus 3pt minus 3 pt}

\setlength{\abovecaptionskip}{0pt}
\setlength{\belowcaptionskip}{-5pt}

\setlength{\abovedisplayskip}{3pt}
\setlength{\belowdisplayskip}{3pt}
\section{Introduction}\label{sect:intro}
Heterogeneous information networks (HINs) comprise various types of nodes and edges, enabling them to effectively model real-world scenarios, including bibliographic, e-commerce, and social networks. As shown in Fig. \ref{fig:expofhin}, in a bibliographic network, nodes can represent entities such as authors, papers, topics, and venues. The relationships among these nodes include authors writing papers, papers having topics, and papers being published in venues. The structural representation of potential relationships between node types is referred to as a schema, as illustrated in Fig. \ref{fig:expofschema}.

Analysing HIN data is fundamental yet critical. Among various analytic techniques, HIN communities search, i.e., discovering correlated nodes of a given query node, with the same node type (target type), has attracted great attention. The HIN community search has a myriad of applications, such as recommendation~\cite{yu2014personalized}, team formation~\cite{lappas2009finding}, and identification of protein functions~\cite{dittrich2008identifying}. Most of the HIN community search methods~\cite{fang2020effective,jiang2022effective,yang2020effective} apply the meta-path ($\mathcal{P}$)~\cite{sun2011pathsim} to first establish virtual edges ($\mP$-pairs) over target nodes, which forms a virtual graph ($G_{P}$) based on $\mP$-pairs. This step is necessary since most HINs focus on recording relations between nodes of different types (i.e., no relations exist between nodes of the same type). Then, classical cohesive subgraph models such as $k$-core~\cite{seidman1983network,malliaros2020core} and $k$-truss~\cite{cohen2008trusses,wang2012truss} are applied to the virtual graph, which leads to popular HIN community models known as $(k,\mathcal{P})$-core and $(k,\mathcal{P})$-truss. 
What is more, the $(k,\mathcal{P})$-core model is then extended by considering other cohesiveness, such as spatial~\cite{fang2017effective,al2020topic}, influence~\cite{zhou2023influential}, etc., apart from target nodes' cohesiveness. 

\begin{figure}[t]
    \centering
    \subfigure[An HIN]{
        \label{fig:expofhin}
        \includegraphics[width=.355\textwidth]{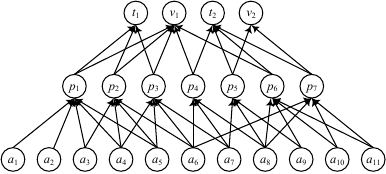}
    }
    \subfigure[Schema]{
        \label{fig:expofschema}
        \includegraphics[width=.095\textwidth]{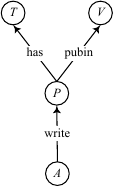}
    }
    \caption{An example of HIN with its schema}
    \label{fig:dblp}
\end{figure}

\begin{figure}[t]
  \centering
  \includegraphics[width=0.7\linewidth]{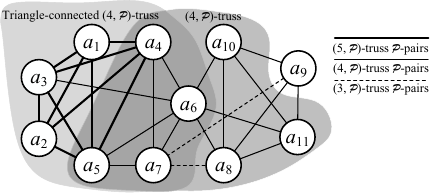}
  \caption{$G_P$ for Fig. 1 with $\mP=(A, P, A)$}
  \label{fig:expofGp}
\end{figure}
However, in many real applications, it is often necessary to impose a size constraint on the identified community due to budget or resource limitations. This motivates us to study the problem of community search in HIN with a size constraint. To the best of our knowledge, there is no existing work that focuses on this problem on HIN data.

\noindent\textbf{Our community model}. We first introduce a refined $(k,\mP)$-truss model, which serves as a cohesiveness measure. Then, we discuss the size constraint and the optimisation objective.

\noindent\textit{\underline{Triangle connected $(k,\mP)$-truss}}. We adapt the mentioned $(k,\mP)$-truss to capture the cohesiveness of the target community. Given a query meta-path $\mP$, if there is an instance of $\mP$ connecting two target nodes, then the two nodes are called a $\mP$-pair. Given three target nodes, they form a triangle if any two of the nodes are in a $\mP$-pair. A set of $\mP$-pairs is called a $(k,\mP)$-truss if every $\mP$-pair in the set forms at least $k$-2 triangles with other $\mP$-pairs in the set.
{\color{orchid}Based on the definition, given a query node $q$, it can only be contained in one single connected community. This community may be formed by multiple dense parts, where some of these parts share exactly one node in common, i.e., they are `weakly' connected.} 
Such a structure may not be ideal for the community search problem. Therefore, to further enhance the effectiveness of communities searched from a query node, we introduce the triangle connectivity into the $(k,\mP)$-truss model~\cite{yang2020effective}. For a set of $\mP$-pairs that is a $(k,\mP)$-truss, it further enforces that every two triangles either share one $\mP$-pair or can be reached by a sequence of triangles s.t. the adjacent triangles in the sequence share one $\mP$-pair. By doing so, one query node can be part of multiple communities, and the community could be denser, making the community model more effective. 

\noindent\textit{\underline{Size constraint and optimisation objective}}. As discussed, the existing community search problems dealing with HIN data do not consider the size constraint that naturally exists in real applications, and they are likely to maximise the community size. In contrast, we introduce a size constraint $s$ to limit the nodes that form the triangle connected $(k,\mP)$-truss. Accordingly, we take cohesiveness as the optimisation objective, i.e., we aim to find the maximum possible $k^{*}$ such that there exists a triangle connected $(k^{*}, \mathcal {P})$-truss containing the query node. We shorten our studied problem, i.e., the \textbf{S}ize-constraint \textbf{C}ommunity \textbf{S}earch in \textbf{H}IN maximising cohesiveness problem, as the SCSH problem. 

\noindent\textbf{Applications.} As discussed, the SCSH problem has many real-world applications. Below, we present two of them.

\begin{itemize}
    \item To organise an academic seminar within Fig.~\ref{fig:expofhin} and invite a key expert $a_6$, it is desirable that participants have direct collaborations, which can be modelled by the meta-path $\mP = (\textit{A, P, A})$. Derived by $\mP$, the $G_P$ is shown in Fig.~\ref{fig:expofGp}. The node set $\{a_i \mid 1 \le i \le 11\}$ forms a triangle-connected $(3, \mP)$-truss. If the venue can accommodate at most 7 participants, there are two possible communities of size 7: $\{a_1, a_2, a_3, a_4, a_5, a_6, a_7\}$ with triangle-connected trussness of 4 and $\{a_4, a_5, a_6, a_7, a_8, a_{10}, a_{11}\}$ with triangle-connected trussness of 3 (even with trussness of 4). Therefore, the organiser will invite the first group, $\{a_1, a_2, a_3, a_4, a_5, a_6, a_7\}$, to the seminar.

    \item When a user searches for product $q$, the platform can leverage the meta-path  $\mP = (\textit{Product, User, Review, User, Product})$ to identify a set of products containing $q$ that are semantically and closely related in terms of $\mP$. The platform then recommends products from this set. However, recommending too many items may overwhelm the user. Therefore, limiting the number of highly relevant products simplifies decision-making and enhances the overall user experience.
\end{itemize}

Our model identifies tightly connected, size-bounded communities for such scenarios.


\noindent
\textbf{Technical challenges and state-of-the-art}. We prove that the SCSH problem is NP-hard, which is indeed very different from the well-studied $(k,\mP)$-truss problem without a size constraint. 
There are related works for homogeneous networks addressing the most cohesive $k$-core~\cite{yao2021efficient} and $k$-truss~\cite{zhang2023size} problems. 
Both approaches solve the problems using an inclusion-and-exclusion-based node set generation framework equipped with corresponding optimisations such as branching, bounding, and candidate set reduction techniques. \cite{yao2021efficient} is advanced in terms of devising upper bounds, as the node set generation-based framework naturally supports cohesiveness defined on nodes more effectively. Due to the difficult of devising effective upper bounds, \cite{zhang2023size} is more advanced in terms of devising lower bounds. 

\noindent\textbf{Technical gaps}. Although we can get ideas from \cite{yao2021efficient} and \cite{zhang2023size}, we find the following technical gaps. First, all of them are dedicated to homogeneous networks. In fact, HINs possess certain desirable properties that can accelerate the search, which will be explored in this paper. Second, all existing works use an inclusion-and-exclusion-based node set generation framework, which is efficient for generating all combinations of $n$ nodes. However, our problem only requires exploring $s$-size combinations of $n$ nodes, and practically, $s$ is far less than $n$. The inclusion-and-exclusion-based node set generation leads to a recursion depth of $n$ rather than $s$. We need a more efficient way of doing so. Last but not least, even the most closely related work \cite{zhang2023size} does not consider triangle connectivity when devising its techniques, which could be improved.     

\noindent\textbf{Our solution}. To fill the above research gaps, we present a novel branch-and-bound algorithm for the SCSH problem based on the lexicographic generation framework. This naturally bounds the recursion depth in terms of $s$. Such an improvement allows us to attempt both the edge set and the node set enumeration-based approaches.  
In terms of bounding techniques, we first take advantage of the HIN property, i.e., taking polynomial time to get large pairwise connected $\mP$-pairs, to derive near-optimum results, serving as a lower bound. 
This lower bound can greatly reduce the candidate size. We then devise novel and effective upper bounds for both node set and edge set enumeration approaches, taking into account triangle connectivity. 
For the edge set enumeration, we propose an upper bound by enlarging the trussness of the edge in the partial result with the minimum trussness, considering the maximum triangle connected triangles that can be included from the candidate set. 
{\color{orchid} However, the above upper bound cannot be applied to the node set enumeration approach. 
This is because, for node set enumeration, we can remove some edges with low trussness from the node set induced subgraph to improve the result, i.e., deriving an upper bound on edges with minimum trussness does not make sense here.} 
We alternatively seek to increase the trussness of the high trussness edges in the partial result, in order to satisfy the necessary condition for improving the overall trussness and to derive a reasonable upper bound. Using the lower and upper bounds, we can effectively and efficiently bound the search towards the optimum result. Furthermore, we propose a novel total search order and branching strategy. Together with all these optimisations, the edge set generation based algorithm runs in $\mathcal{O}(|N(q,G_{P})|\binom{|E(G)|}{\frac{(s-1)(s-2)}{2}}|E(G_{P})|^{1.5})$  with the most advanced upper bounds and the node set generation based algorithm runs in $\mathcal{O}(|N(q,G_{P})|\binom{|V(G)|}{s-2}|E(G_{P})|^{1.5})$ with a reasonable upper bound, where $G_{P}$ is the virtual graph consisting of all $\mP$-pairs in HIN. 
In practice, the node set generation approach runs one order of magnitude faster than the edge set generation approach, and the edge set generation approach is, on average, 40\% faster than the method adapted from SOTA. 


Our principal contributions are summarized as follows.
\begin{itemize}
    \item We introduce the problem of size-bounded community search in heterogeneous information networks, maximising the cohesiveness. \hfill{\textbf{Section}~\ref{sec:def}}
    \item We introduce a novel search framework that brings low time complexity for both edge and node set enumeration based methods. \hfill{\textbf{Section}~\ref{sec:frame}}
    \item We develop a novel heuristic algorithm that efficiently identifies a feasible community, serving as a tighter lower bound. \hfill{\textbf{Section}~\ref{sec:lb}}
    \item We propose novel bounding, branching, total order, and reduction optimisations, which lead to two efficient exact algorithms.
    
    \hfill{\textbf{Sections}~\ref{sec:bound} to \ref{sec:reduction}}
    
    \item We validate our approaches through extensive experiments on real-world datasets, demonstrating their effectiveness and efficiency. \hfill{\textbf{Section}~\ref{sec:exp}}
\end{itemize}

\section{Problem Formulation}
\label{sec:def}
\subsection{Preliminaries}

\begin{definition}
	\textbf{Heterogeneous Information Network (HIN):} An HIN is a directed graph $H(V,E)$ with a node type mapping function $\phi : V\rightarrow \mathcal{A}$ and an edge type mapping function $\psi : E\rightarrow \mathcal{R}$, in which every $v\in V$ has a node type $\phi(v)\in \mathcal{A}$ and every $e\in E$ has an edge type $\psi(e)\in \mathcal{R}$.
	\label{HINdef}
\end{definition}

\begin{definition}
	\textbf{HIN Schema:} Given an HIN $H=(V, E)$ with mappings $\phi:V \rightarrow \mathcal{A}$ and $\psi: E \rightarrow \mathcal{R}$, its schema is a graph defined over node types $\mathcal{A}$ and edge types $\mathcal{R}$, denoted as $T_H(\mathcal{A},\mathcal{R})$.
	\label{Schemadef}
\end{definition}


Fig.~\ref{fig:expofschema} shows the DBLP schema with node types $A$, $P$, $T$, and $V$ (authors, papers, topics, venues) and where edges are reversible (e.g., $A$–$P$ denotes ``write'' while $P$–$A$ denotes ``written by'').

\begin{definition}
	\textbf{Meta-path:} A meta-path $\mP$ is a path defined on a schema $T_H=(\mathcal{A},\mathcal{R})$, and is denoted in the form $A_1\stackrel{R_1}{\longrightarrow} A_2\stackrel{R_2}{\longrightarrow}\cdot \cdot \cdot \stackrel{R_l}{\longrightarrow} A_{l+1}$, where $l$ is the length of $\mP$, $A_i\in \mathcal{A}$, and $R_i\in \mathcal{R}(1\leq i\leq l)$. For simplicity, we denote $\mP$ as an ordered set in forms of $(A_1, A_2, \ldots,  A_{l+1})$.
	\label{Metapathdef}
\end{definition}

Given that relationships in an HIN are reversible, meta-paths are thus reversible. If reversed meta-path $\mP^{-1}$ equals $\mP$, it is symmetric (e.g., $(A, P, A)$). In this paper, all meta-paths are assumed symmetric.


\begin{definition}
	\textbf{Instance of meta-path:} Given a path $p=a_1\rightarrow a_2\rightarrow \cdot \cdot \cdot \rightarrow a_{l+1}$ and a meta-path $\mP=(A_1, A_2, \ldots, A_{l+1})$, $p$ is an instance of $\mP$, if $\forall i\in [1,l]$, the node $a_i$ and edge $e_i=(a_i,a_{i+1})$ satisfy $\phi(a_i)=A_i$ and $\psi(e_i)=R_i$.
	\label{Instancedef}
\end{definition}

\noindent\textbf{$\mP$-pair and $G_{P}$}. The source and sink nodes of a $\mP$ instance form a $\mP$-pair. For example, in Fig. \ref{fig:expofhin}, the path $a_1 \rightarrow p_1 \rightarrow a_4$ is an instance of the meta-path $(A, P, A)$, and $(a_1, a_4)$ is a $\mP$-pair. 
After identifying all $\mP$-pairs, we can construct \textbf{a derived homogeneous graph, denoted as $G_P$}. Fig. \ref{fig:expofGp} illustrates an example of $G_P$ derived from Fig. \ref{fig:expofhin} based on the meta-path $(A, P, A)$.


\begin{definition}
    \textbf{Triangle ($b$-triangle in~\cite{yang2020effective}):} Given an HIN $H$ and a meta-path $\mP$, a triangle $\bigtriangleup$ is a triplet of nodes $(u, v, w)$ such that every pair of nodes in $\bigtriangleup$ forms a $\mP$-pair.
    \label{def:btriangle}
\end{definition}

In \cite{yang2020effective}, two types of triangles are discussed according to the constraints over the instances of $\mP$ forming a triangle. In this paper, we focus more on $b$-triangles and will call them triangles from now on for brevity.  


\begin{definition}
    \textbf{Support:} Given a set $S$ of $\mP$-pairs for $H$, the support of a $\mP$-pair $(u,v)$ regarding $S$ is the number of triangles, which are formed by $\mP$-pairs in $S$, containing $(u, v)$.
    \label{def:bsupport}
\end{definition}

\begin{definition}
    \textbf{($k$, $\mP$)-truss~\cite{yang2020effective}:} Given a set $S$ of $\mP$-pairs for $H$, the ($k$, $\mP$)-truss is the maximum set of $\mP$-pairs, denoted by $S'$, such that for each $\mP$-pair, its support within $S'$ is at least $k-2$.
    \label{def:kPbtruss}
\end{definition}

\begin{definition}
    \textbf{Trussness of $\mP$-pair:} The trussness of a $\mP$-pair $(u,v)$ regarding $S$, denoted by $\tau_S(u,v)$, is the largest $k$ such that $\mP$-pair $(u,v)$ is contained in a ($k$, $\mP$)-truss of $S$.
\end{definition}

\begin{definition}
    \textbf{Trussness of node:} The trussness of a node $v$ regarding $S$, denoted by $\tau_{S}(v)$, is defined as the maximum trussness among all $\mP$-pairs containing $v$, i.e., $\tau_{S}(v) = \max\{\tau_{S}(u,v) \mid (u,v) \in S\}$.
    \label{def:vtruss}
\end{definition}

\subsection{Problem formulation}
Based on the existing definitions for $(k, \mP)$-truss, we further introduce the well-studied concepts of triangle adjacency and triangle connectivity to enhance the efficacy of $(k, \mP)$-truss and then introduce the size constraints, which lead to our studied SCSH problem.   

\begin{definition} 
    \textbf{Triangle adjacency:} Two triangles, $\bigtriangleup_s$ and $\bigtriangleup_t$, are adjacent if they share a $\mP$-pair, i.e., $|\bigtriangleup_s \cap \bigtriangleup_t| = 2$.
    \label{Edgeadjacentdef} 
\end{definition}

\noindent\textbf{Triangle connectivity}. Two triangles, $\bigtriangleup_s$ and $\bigtriangleup_t$, are triangle connected, if there exists a sequence of triangles in $S$, $\bigtriangleup_s, \dots, \bigtriangleup_n$ ($n \geq 2$) such that for $1 \leq i < n$, $|\bigtriangleup_i \cap \bigtriangleup_{i+1}| = 2$ and $\bigtriangleup_1= \bigtriangleup_s$, $\bigtriangleup_n= \bigtriangleup_t$. 

\begin{definition} 
    \textbf{Triangle-connected ($k$, $\mP$)-truss:} A set of $\mP$-pairs $S$ is a triangle-connected ($k$, $\mP$)-truss if it satisfies: (1) $S$ is a ($k$, $\mP$)-truss, and (2) $\forall$ $\bigtriangleup_1$ and $\bigtriangleup_2$ $\in$ $S$, $\bigtriangleup_1$ is triangle connected with $\bigtriangleup_2$  within $S$. \label{ktrusscommunitydef}
\end{definition}

\begin{table}[t]
\caption{Frequently used notations}
\label{notion}
\small
\centering
\begin{tabular}{
  >{\centering\arraybackslash}m{1.2cm}|
  >{\raggedright\arraybackslash}m{6.5cm}
}
\toprule
\textbf{Notation} & \textbf{Definition} \\
\midrule
$\mathcal{P}$ & A meta-path \\
$G_{P}$, $G_{P}[\cdot]$& Graph consists of $\mP$-pairs,  nodes (edges) induced subgraph of $G_{P}$ \\
$C$, $R$ & Partial result, candidate set (a set of edges or nodes) \\
$C^*$ & Global optimum results or optimum result up to the time \\
$UB(C, R)$ & The upper bound of triangle-connected trussness for the recursion state $(C, R)$ \\
$C^{\le \tau}$ & The set of edges that trussness is not greater than $\tau$ \\
$E_{min}$ & The minimum trussness among the edges in $C$ \\
$|\triangle_{opt}|$ & The optimistic number of capable triangles \\
$|E_{opt}|$ & The optimistic number of edges \\
$\Phi_E(u)$ & The set of nodes dominated by $u$ based edges \\
$\Phi_\triangle(u)$ & The set of nodes dominated by $u$ based triangles \\
\bottomrule
\end{tabular}
\end{table}



\noindent\textbf{SCSH problem}. Given an HIN $H(V, E)$, a query node $q$, a community constraint size $s\in \mathbf{N}^{+}$, and a meta-path $\mP$, i.e. the input is $(H, \mP, q, s)$. The objective is to find a community $G_{P}^{\prime}=(V^{*}, S^{*})$ (formed by a set of $\mP$-pairs), which satisfies the following conditions:
\begin{enumerate}
    \item Query awareness: $q\in V^{*}$;
    \item Size constraint: $|V^{*}|=s$;
    \item Cohesiveness maximisation: $G_{P}^{\prime}$ is a triangle-connected ($k$, $\mP$)-truss and has the maximum $k$ compared to any other $G_{P}^{\prime\prime}$ meeting the above two conditions.
\end{enumerate}

An example has been provided in the introduction. And the notations that are frequently used are summarised in Table~\ref{notion}. Next, we analyse the hardness.


\setcounter{theorem}{0}
\begin{theorem}
    SCSH problem is NP-hard.
\end{theorem}

\begin{proof}

We prove that the decision version of the SCSH problem is NP-complete via a reduction from the $m$-clique problem, which is known to be NP-complete.

The $m$-clique problem takes a graph $G=(V, E)$ and an integer $m$ as input, and asks whether $G$ contains a clique of size at least $m$. The decision version of the SCSH problem changes the optimisation objective (maximising the trussness) to a constraint, i.e., $C$ is a triangle-connected $(k, \mP)$-truss for a given $k$ (other constraints are identical to the SCSH problem). 

We first present a gadget that can transform $G$ into an instance of HIN ($H$) with two types of nodes ($T_1$ and $T_2$) and a network schema consisting of just one undirected edge $(T_1, T_2)$. The gadget is as follows. For each $e_i(u,v)$ in $E(G)$ ($1\le i\le |E(G)|$), we create node $u$, $v$ with type $T_1$ in $H$ if $u$ and $v$ are not in $H$ yet (we do not create an edge between $u$ and $v$ in $H$). Then, we create a new dummy node $w_i$ with type $T_2$ for $e_i$. After, we create two undirected edges $(u, w_i)$ and $(v, w_i)$ in $H$. At last, we create an extra dummy node $q$ of type $T_1$ and link it to each created $w_i$. Clearly, applying the gadget takes linear time w.r.t. $|E(G)|$. 

By applying the gadget, the instance of our problem is as follows. $H$ is generated by the gadget, $\mP$ is $(T_1$, $T_2$, $T_1)$, $s$ and $k$ are set to be $m+1$ (notice that a size $s$ $s$-truss is a $s$-clique), and $q$ is query node.

It is easy to see that $\mP$-pairs form a graph $G_{P}$ that contains the identical structure as $G$ in the $m$-clique problem, with an extra node $q$ connecting to every node in $G$. This means that each size $m$-clique in $G$ corresponds to an $m+1$-clique in $G_{P}$. In reverse, every $m+1$-clique in $G_{P}$ (containing $q$) leads to a $m$-clique in $G$ by removing $q$. 
It is clear to see that computing $G_{P}$ for the given $H$ takes $p$-time. 

As such, we have shown that the above instance of the decision version of our studied problem is at least as hard as the $m$-clique problem, i.e., the decision version of our problem is NP-complete, which means that the SCSH problem is NP-hard.   

\end{proof}

\section{State-of-the-art framework}

Size-bounded community search maximising cohesiveness has been explored extensively in homogeneous networks, particularly for the $k$-core and $k$-truss models. 
We first retrace the methodological development of existing approaches and then discuss the rationales for the algorithms proposed for our problem.

\noindent\textbf{SOTA framework~\cite{yao2021efficient,zhang2023size}}. Interestingly, SOTAs use include-exclude backing tracking to systematically explore node sets that potentially lead to the desired result, as shown in Algorithm~\ref{alg:sota}, which is decorated with upper/lower bound-based pruning, candidate reductions, and branching strategies.

\noindent\textit{\underline{Upper/lower bound-based bounding}}. 
Whenever Algorithm~\ref{alg:sota} enters a new recursion state, it attempts to stop exploring non-promising recursion states by estimating an upper bound on $(C,$  $R)$, i.e., the highest cohesiveness, and if the upper bound is no better than the lower bound ($\tilde{k}$), the current recursion state and any other recursion states induced by the current recursion state cannot lead to a result with cohesiveness greater than $\tilde{k}$. 
Notice that the upper bound is based on both $C$ and $R$, which secures the correctness of the bounding (we will revisit the detail later).   
As such, Algorithm~\ref{alg:sota} only explores recursion states that cannot be bounded.

\noindent\textit{\underline{Heuristic during search}}. SOTAs propose corresponding scores for nodes in the candidate set $R$, which ensures that nodes in $R$ contribute less to the cohesiveness of $C$ as recursion goes deep (line~\ref{a1:score}). This enables the search to produce a result with high cohesiveness at early search stages, which enhances the upper/lower bound-based bounding.

\noindent\textit{\underline{Eager candidate reductions}}. SOTAs also incorporate candidate reductions based on the coreness (or trussness), distance, etc. 
The reduction is applied whenever Algorithm~\ref{alg:sota} enters a new recursion state (line~\ref{a1:red}).

\noindent\textit{\underline{Branching strategy}}. SOTAs propose dominance-based branching rules. 
The key idea is that, given two nodes $u$ and $v$, if $u$ dominates $v$, then a higher cohesiveness result must be induced by a recursion state including both $u$ and $v$ or a recursion state excluding $v$. 
This can be extended to when $u$ dominates multiple nodes (lines \ref{a1:dos} to \ref{a1:doe}).   


\begin{algorithm}[t]    
    \tcc{Initially $C=\{q\}$, $R$ contains nodes reachable from $q$, $\tilde{k}$ is a heuristic result as a lower bound}
    Reduce $(C, R)$ \label{a1:red}\;
    \If{$|C|=s$}{
        $k\gets $ cohesiveness of $C$; \tcp{trussness or coreness}
        \lIf{$k>\tilde{k}$}{
            $\tilde{k}\gets k$,
            $optC\gets C$
        }
    }
    \If{$|C|<s$ $\&\&$ $R\ne \emptyset$ $\&\&$ $UB(C,R)> \tilde{k}$}{
        $v\gets$ pop a node with the highest score from $R$ \label{a1:score} \;
        $\Phi$ $\gets$ compute node dominated by $v$ \label{a1:dos}\;
        \ForEach{$u\in \Phi$}{
            $X$ $\gets$ $X\cup\{u\}$; \tcp{$X$ is $\emptyset$ initially}
            B\&B($C\cup \{v,u\}$, $R$ $\setminus X$) \label{a1:doe}\;
        }
        B\&B($C\cup\{v\},R\setminus  X$)\;
        B\&B($C,R\setminus X$)\;
    }
\caption{B\&B($C$,$R$)}
\label{alg:sota}
\end{algorithm}

\noindent\textbf{Opportunities}. Based on the above discussions, we identify the following opportunities, bearing SOTAs in mind.  

First, although STOAs have made great efforts to prune unpromising recurrences, they only explore homogeneous graph properties. We explore HIN properties, which lead to an effective heuristic algorithm that can quickly find a near-optimal result. 
This significantly enhances pruning effectiveness and reduces overall runtime.  

Second, we introduce the $k$-combination recursion-based B\&B framework (\textsf{kcB\&B}) for solving the size constraint community search maximising cohesiveness. 
Compared to the include–exclude based framework runs in $\mathcal{O}^{*}(2^{n})$ regardless the size constraint $s$, \textsf{kcB\&B} naturally runs in $\mathcal{O}^{*}(C_{n}^{s-1})$. 
Enhanced by other techniques proposed by us, $\mathcal{O}^{*}(C_{n}^{s-1})$ can be further reduced to $\mathcal{O}^{*}$ $(n^{\prime} C_{n}^{s-2})$, where $n^{\prime}$ can be much less than $n$.   

Last but not least, compared to the $k$-core upper bound proposed in \cite{yao2021efficient}, SOTA for $k$-truss pays less attention to the upper bound while focusing more on the lower bound, leaving room for improvement. 

In the next section, we propose our algorithms with novel optimisations based on the aforementioned opportunities. 
\section{Our Methods}
In this section, we introduce the lexicographic generation B\&B framework first, which serves as the backbone of our proposed methods. 
Then, we discuss our main optimisations, including novel heuristics, upper bounds and total search order in great detail. 
Then, we concisely discuss and adapt several existing reduction rules to our framework. 

\subsection{Lexicographic generation framework}\label{sec:frame}
\begin{algorithm}[t]

    $C\gets \{q\}$; $C^{*}$ $\gets$ the best heuristic result;\tcp{Section~\ref{sec:lb}}
    $R \gets$ target nodes reachable from $q$ via $\mP$-pairs\;
    reduce $R$ according to $\tau_{\triangle}(C^{*})$ and order the remaining $R$;\tcp{Section~\ref{sec:reduction}}
    \ForEach{$v\in R$}{ 
        \tcp{Section~\ref{sec:order}}
        \If{stop conditions cannot be satisfied }{
            \textsf{kcB\&B}$(C\cup\{v\},R\setminus\{v\})$\; 
            $R$ $\gets$ $R\setminus\{v\}$\;
        }
    }
    \Return $C^{*}$\;
    
    \BlankLine 
    \SetKwFunction{FMain}{kcB\&B} 
    \SetKwProg{Fn}{Function}{:}{\KwRet} \Fn{\FMain{$C,R$}}{ 
        \If{$|C|==s$}{
            \lIf{$\tau_{\triangle}(C)>\tau_{\triangle}(C^{*})$}{
                $C^{*}$ $\gets$ $C$
            }
        }
        \tcp{Section~\ref{sec:bound}}
        \If{$|C|<s$ $\&\&$ $|R|\ge s-|C|$ $\&\&$ $UB(C,R)>\tau_{\triangle}(C^{*})$}{
        \tcp{Section~\ref{sec:branch}}
            $v$ $\gets$ the most dominating node in $R$\; 
            $\Phi$ $\gets$ nodes in $R$ dominated by $v$\; 
            \ForEach{$u \in R\setminus \Phi$}{
                \textsf{kcB\&B}$(C\cup \{u\},R\setminus\{u\})$; $R$ $\gets$ $R\setminus \{u\}$\;
            }
        }
    }
\caption{The new framework for SCSH}
\label{alg:nf}
\end{algorithm}

Due to the NP-hardness of our studied problem, no polynomial algorithm exists for solving it.  
As shown in SOTAs, we have to systematically try combinations of target nodes of size $s$ and see if the target nodes induced $\mP$-pairs (target nodes induced subgraph in $G_{P}$) can form a triangle-connected $(k,\mP)$-truss with node size $s$ while maximising the cohesiveness (trussness). 

We find that the lexicographic generation~\cite{Knuth2011Art} is the most efficient way to generate $s$-node sets from $V(G_{P})$, especially when $s$ is small, which nicely fits our application scenario. 
The reason is that lexicographic generation explores a generation tree with depth up to $s$ to generate all $s$-node sets from $|V(G_{P})|$. 
In contrast, the inclusion-and-exclusion-based method explores a generation tree with depth up to $|V(G_{P})|$. Due to the differences in maximum recursion depth, the lexicographic generation naturally exhibits better scalability.  

To explain the algorithm clearly, we first introduce a concept.

\begin{definition}\label{def:tSet}\textbf{Triangle connected trussness of a target node set}.
Given a target node set $C$, the triangle connected trussness of $C$ is the maximum $k$ such that there exists a triangle connected $(k, \mP)$-truss in $C$ induced subgraph of $G_{P}$, i.e., $G_{P}[C]$, where the triangle connected $(k, \mP)$-truss contains every node in $C$, denoted by $\tau_{\triangle}(C)$. 
\end{definition}

Below, we call `triangle connected trussness' as `trussness' for brevity. Besides, under the context of $G_{P}$, trussness in $G_{P}$ is the same as trussness of $\mP$-pairs.   

Algorithm~\ref{alg:nf} shows our new search framework based on the lexicographic generation. The recursive function \textsf{kcB\&B} serves as the backbone, systematically exploring $s$-node sets. 
Algorithm~\ref{alg:nf} first obtains a heuristic result that serves as an initial (trussness) lower bound (line 1). 
It then initialises the result set $C$ by including the query node $q$ and collects the complete candidates ($R$), i.e., target nodes that are $\mP$-pair reachable from $q$. 
It leaves room for ordering optimisations, which can apply novel techniques that can lead to potential early terminations (line 5). 
Then, \textsf{kcB\&B} is called to explore $s$ node sets from $R$, which only contains target nodes potentially leading to better results (line 3).  
When \textsf{kcB\&B} enters a recursion state with $(C,R)$ having $|C|$ less than $s$, it attempts to stop and backtrack to other more promising recursion states by evaluating two bounding rules (line 12). The first size-based bounding ($|R|\ge s-|C|$) is trivial, while the latter, i.e., deriving an upper bound based on both $C$ and $R$ and using it to compare to the best result up to the time, will be discussed later. 
When \textsf{kcB\&B} reaches to $s$-node set (line 10), it derives $\tau_{\triangle}(C)$ and updates $C^{*}$ if necessary. 

The completeness of Algorithm~\ref{alg:nf} is clear since it is based on the lexicographic generation. Suppose all the reductions, bounding, branching rules and early terminations are correct (which we will discuss and prove in great detail), $C^{*}$ is optimum when Algorithm~\ref{alg:nf} terminates. Ignoring optimisations to be discussed, Algorithm~\ref{alg:nf} runs in $\mathcal{O}^{*}(\binom{|V(G)|}{s-1})$.

\noindent\textbf{Edge sets generation (ESG) based approach}. The backbone of Algorithm~\ref{alg:nf} is to try node sets of size $s$, i.e., node sets generation (NSG), which is straightforward by the definition of our problem. 
However, as shown in Definition~\ref{def:tSet}, the result is not a node-induced subgraph of $G_{P}$. 
This indicates an alternative to Algorithm~\ref{alg:nf}, i.e., we try edge sets till the node size reaches $s$. 
In fact, this approach is more natural for truss-based cohesiveness since trussness is defined on edges, and the result is an edge-induced subgraph. 
We also attempt this approach and discover several interesting experimental results. Since this approach shares great similarities with Algorithm~\ref{alg:nf}, we omit its details to avoid duplication. 

Below, we start to discuss detailed optimisations. Unless explicitly discussed, optimisations apply to both Algorithm~\ref{alg:nf} (NSG) and ESG, with NSG being the default context.

\subsection{Lower bound initialization}\label{sec:lb}
One of the critical speedups for Algorithm~\ref{alg:nf} is to have a near-optimal result, i.e., a tight initial (trussness) lower bound. This is because a near-optimal result can be utilised to reduce the candidate set and bound the search. In this subsection, we utilise HINs to propose a polynomial algorithm that effectively initialises a tight lower bound.  

\noindent\textbf{Rationale}. We can directly design heuristics and greedy algorithms to generate $s$ node sets by including nodes with high trussness and see if we can get a quality result. The challenge is how we can evaluate the quality. As such, we first discuss the optimality of the lower bound for our problem. We then demonstrate that the quality lower bounds can be derived efficiently by leveraging the properties of HINs. 


\begin{algorithm}[t]
    \caption{SCSHHeu($H$,$\mP$,$q$,$s$)}
    \label{alg:heu}
    \small
    $\mathcal{S}$ $\gets$ the set of seed cliques;\tcp{$\mP$-stars containing $q$}
    \eIf{$|\argmax_S\{|S||S\in \mathcal{S}\}|\ge s$}{
        $\tilde{C}\gets$ any $s$ node set of $S$\ that contains $q$;
    }
    {
        $\tilde{C}$ $\gets$ $\emptyset$\;
        sort $\mathcal{S}$ in non-increasing order according to $|S|$ ($S\in \mathcal{S}$)\;
        \tcp{We can terminate the loop after a few attempts}
        \While{$\mathcal{S}\ne \emptyset$}{
                $C^{\prime}$ $\gets$  $\mathcal{S}$.pop()\; $\mathcal{S}^{\prime}$ $\gets$ $\mathcal{S}$ $\cup$ $\mP$-stars not containing $q$ but $\frac{\mP}{2}$ reachable from $q$\; 
                
                \While{$\mathcal{S}^{\prime}\ne \emptyset$}{
                    $C^{\prime\prime}$ $\gets$ $\argmax_{S}\{|E(G_{P}[S\cup C^{\prime}])||S\in \mathcal{S}^{\prime} \}$; $\mathcal{S}^{\prime}$ $\gets$ $\mathcal{S}^{\prime}\setminus \{C^{\prime\prime}\}$;
                    $C^{\prime}$ $\gets$ $C^{\prime} \cup C^{\prime\prime}$\;
                     \If{$|C^{\prime}|\ge s$}{
                        reduce $C^{\prime}$ to size $s$ by progressively removing edges with the lowest trussness\;
                        \lIf{$q\in C^{\prime}$ \&\& $\tau_{\triangle}(\tilde{C})$ $<$ $\tau_{\triangle}     (C^{\prime})$}{
                            $\tilde{C}$ $\gets$ $C^{\prime}$
                        }
                     }
                    
                }
        
        }
    }
    \Return $\tilde{C}$\;
\end{algorithm}

\noindent\textbf{Optimum lower bound}. 
The clique of $G_{P}$ containing $q$ provides a lower bound for the trussness of $G_{P}$, i.e., if there exists a $k$-clique containing $q$ then there exists a triangle-connected truss containing $q$ with trussness at least $k$. The optimum lower bound (denoted by $\underline{k}$) by considering the size constraint is shown as follows. 

\begin{align}\label{eq:lb}
     \underline{k} = s & \quad \text{if $\omega_{q}(G_{P})\ge s$},
\end{align}
where $\omega_{q}(G_{P})$ denotes the maximum clique size in $G_{P}$ containing $q$.

Based on the above definition, to get a good lower bound, we want to find large cliques containing $q$ in $G_{P}$. 
However, this could be hard, especially if we treat $G_{P}$ as a homogeneous graph and ignore the fact that $G_{P}$ is derived via instances of meta-path. 
To show why finding large-sized cliques in $G_{P}$ is easy, we show an existing definition below.

\begin{definition}
    \textbf{$\mP$-star}\cite{yang2020effective}: Given an HIN $H(V,E)$ and a symmetrical meta-path $\mP=(A_1, A_2, ..., A_{l+1})$, a set of nodes with type $A_1$ is a $\mP$-star if there is a node $a_{\frac{l+1}{2}}$ with type $A_{\frac{l+1}{2}}$ s.t. for each $u$ in this set, $a_{\frac{l+1}{2}}$ is a $\frac{\mP}{2}$-neighbor of $u$, i.e., they are $\frac{\mP}{2}$ reachable.
    \label{def:pstar}
\end{definition}

\begin{obv}\label{ob:star}
    Every $\mP$-star is a clique in $G_{P}$.
\end{obv}

Based on \textsc{Observation}~\ref{ob:star}, we can quickly get a set of cliques containing $q$ by simply applying breadth first search on $G$ (HIN) with type $A_{\frac{l+1}{2}}$ nodes that are  $\frac{\mP}{2}$ reachable from $q$. We treat these $\mP$-stars as \textit{seed cliques}. 

The seed cliques could be as large as $s$, as shown in Equation~\ref{eq:lb}, which leads to an optimum result directly, and Algorithm~\ref{alg:nf} does not need to invoke recursive calls. 
However, the seed cliques could be small. 
This motivates us to derive a tighter lower bound based on seed cliques. 

\noindent\textbf{Clique-based lower bound}. We observe that if we union multiple seed cliques with more edges in their induced subgraph of $G_{P}$, the chance of forming high trussness subgraphs is high. Based on this rationale, the complete heuristic algorithm is shown in Algorithm~\ref{alg:heu}. 

If there is no seed clique ($\mP$-star) containing $q$ with size no less than $s$, Algorithm~\ref{alg:heu} initialises $\tilde{C}$ as $\emptyset$ since we do not have any feasible lower bound yet.  
Then, Algorithm~\ref{alg:heu} gives large seed cliques opportunities. For each large seed clique $C^{\prime}$ (line 8), it will be progressively enlarged by greedily merging other seed cliques $C^{\prime\prime}$ that bring the most edges w.r.t. $C^{\prime}$ up to the time. 
After a union, if $C^{\prime}$ is larger than $s$, Algorithm~\ref{alg:heu} reduces it to fit the size and then replaces it with the best lower bound if necessary. 
Notice that Algorithm~\ref{alg:heu} is shown as giving each seed clique a chance; in fact, we can terminate it after a fixed number of rounds. Algorithm~\ref{alg:heu} returns $\tilde{C}$ with $\tau_{\triangle}(\tilde{C})$ as output.  

\noindent\textit{\underline{Computational cost}}. 
The time complexity is $O(c \cdot |E(G_{P}[C'])|^{1.5})$.  
The dominating parts are the two \textit{while} loops together, having a time complexity of $O(|\mathcal{S}|\cdot|\mathcal{S}'||E(G_{P}[C'])|^{1.5})$ (notice we can set a constant $c$ and terminate this nested loop much earlier). The innermost loop is dominated by the cost of computing edge trussness and iteratively deleting edges, which takes $O(|E(G_{P}[C'])|^{1.5})$.

\begin{example}
    As shown in Fig.~\ref{fig:expofhin}, consider $q = a_6$, $\mP = (A, P, A)$, and $s=7$. Let $\mathcal{S} = \{p_2,$ $p_3,$ $p_4,$ $p_7\}$. 
    We denote $\mP$-stars using the centre nodes of the instances of $\mP$. 
    The largest $\mP$-star is $\{a_6, a_8, a_{10}, a_{11}\}$, which is smaller than $s$. Therefore, Algorithm~\ref{alg:heu} proceeds to the ``else'' branch.
    
    After sorting, $\mathcal{S} = \{p_3, p_2, p_7, p_4\}$. 
    We first select $p_3$ as the initial seed clique, i.e., $C' = \{a_4, a_5, a_6, a_7\}$ and update $\mathcal{S}' = \{p_2, p_7, p_4, p_1, p_5,p_6\}$. 
    We then merge $p_1$ into $C'$ since $p_1$ contributes the largest number of edges. 
    As the resulting size is $|C'| = 7$, which satisfies the size constraint, no size reduction is needed. This yields the final result $\{a_1, a_2, a_3, a_4, a_5, a_6, a_7\}$ with trussness of $4$, which is the best result returned by Algorithm~\ref{alg:heu}.
    
\end{example}

\subsection{Bounding techniques}\label{sec:bound}




In this subsection, we present how to utilise a known result (the heuristic result discussed previously), i.e., a (trussness) lower bound, to bound the search tree by comparing it to an upper bound. 

\noindent\textbf{Rationale}. Given a non-trivial recursion state with $C$ and $R$ ($|C|<s$ and $|R|\ge s-|C|$), we want to optimistically estimate the best trussness after moving arbitrary $s-|C|$ nodes from $R$ to $C$, i.e., an upper bound for this recursion state, denoted by $UB(C, R)$. 
Clearly, as long as $UB(C, R)$ is admissible, that is, 1) $UB(C, R)$ is no less than any results\footnote{in terms of trussness} induced by the current recursion state, and 2) $UB(C, R)$ is no greater than the current lower bound, we can stop branching at the current recursion state. Algorithm~\ref{alg:nf} can explore (backtrack to) other recursion states that cannot be bounded. 

As discussed, SOTAs all involve the idea above. In~\cite{yao2021efficient}, the main idea is as follows. It is assumed that 1) $s$-$|C|$ nodes from $R$ that are mostly connected to $C$ would be added to $C$, 2) these nodes all exhibit high coreness, and 3) all the edges incident from these nodes would contribute to nodes in $C$ with the minimum coreness. The process of deriving the upper bound is dynamic, i.e., when all the nodes in $C$ with minimum coreness have been increased by 1, the remaining edges will be added to the minimum coreness nodes until all incident edges are consumed. After that, the updated minimum coreness among all nodes in $C$ serves as the upper bound. 

On the other hand, STOA for truss-based model~\cite{zhang2023size} is mainly based on the current trussness of $C$ and the number of nodes in $R$. It is based on the fact that when moving a node from $R$ to $C$, the trussness of $C$ will be increased by at most $1$. Then, given the trussness of $C$, \cite{zhang2023size} computes how many nodes are needed, and if $R$ cannot provide sufficient nodes, the branch can be bounded. The admissibility of this bound lies in the fact that not every node can increase the trussness.

Clearly, the concept presented in \cite{yao2021efficient} is more refined, as it considers necessary structural conditions based on the cohesive subgraph model. If it can be adapted to our problem, we could attain a much tighter bound compared to \cite{zhang2023size}.

We shall present two non-trivial upper bounds—one for ESG and one for NSG. 
As mentioned, ESG is more natural when deriving the upper bound using the idea of \cite{yao2021efficient}. Our contributions for designing ESG upper bounds are twofold: 1) to address the challenges induced by the cohesiveness difference, and 2) to address the inconsistency between edge inclusion and node inclusion (i.e., adding an edge may not increase the size of $C$). 

Then, we propose a novel upper bound for NSG, which is more challenging since the results of our model are edge-induced graphs, i.e., when NSG generates a result $C$ (a node set), it is not necessarily true that all edges in $G_{P}[C]$ are in the result.

\subsubsection{Upper bounds for ESG} We consider a non-trivial recursion state with $C$ and $R$, 
where both $C$ and $R$ contain edges in $G_{P}$, 
and we use $V(C)$ to denote the nodes contained in $C$. 

\noindent\textbf{ESG upper bound computation}. The ESG upper bound is based on the accurate trussness for edges in $C$ and on how $R$ can increase the trussness of $C$ optimistically. 

The trussness for each edge in $C$ can be progressively updated during the search process and can be treated as known information. 
We can partition edges according to their trussness. 
Our proposed upper bound starts with the minimum trussness edges in $C$ denoted by $E_{min}$ with trussness of $\tau_{\triangle}(C)$, where $\tau_{\triangle}(C)$ is trussness of $C$.


To establish the optimistic while tight upper bounds, we assume edges in $R$ which can contribute to $R$ the most \textit{capable triangles} (c-triangle), i.e., triangles consisting of edges with trussness no less than $\tau_{\triangle}(C)$, will be added to $C$. We also assume that all these capable triangles will contribute to $E_{min}$ up to the time. Additionally, since including an edge may or may not introduce a new node, the optimistic number of c-triangles is derived as follows.   


\noindent\textit{\underline{Edges with no extra nodes}}. For all these edges that form c-triangles with $C$, newly induced c-triangles are assumed to contribute to the trussness of $E_{min}$ up to the time. 

\noindent\textit{\underline{Edges with extra nodes}}. For edges with extra nodes, for each of such extra node $u$, we compute its score in terms of the number of c-triangles that $u$ contributes to $C$ (at least one of the nodes of such c-triangle is in $V(C)$) in the $C\cup R$ induced subgraph of $G_{P}$. For the top $s-|V(C)|$ of the nodes with the highest score, they are assumed to move to $C$ and contribute to the support of $E_{min}$ up to the time.

For these optimistically estimated c-triangles assumed to be added to $C$, each edge in $E_{min}$ will consume one of the c-triangles up to the time. 
When the support of every edge in $E_{min}$ has been increased by 1, the estimated upper bound is increased by 1. 
Then $E_{min}$ is updated by including edges with minimum trussness up to the time. 
We repeat this process until all the optimistically estimated c-triangles have been consumed.  

\noindent\textit{\underline{Admissibility}}. The above process naturally follows the trussness definition to optimistically estimate the trussness of enlarged $C$. In fact, 1) the number of c-triangles assumed to be added is an overestimation, 2) not all the c-triangles happen to contribute to the edge with the minimum trussness, and 3) after enlarging $C$, $C$ contains more edges while these edges are assumed to have high trussness. 

\noindent\textit{\underline{Computational cost}}. 
The dominant cost of computing the upper bound is the number of triangles in $C\cup R$ induced subgraph of $G_{P}$ (edge induced subgraph), which is bounded by $\mathcal{O}(|C\cup R|^{1.5})$. 

\noindent\textbf{Grouped ESG upper bound computation}. Given two edge sets $C$ and $C'$ where $C'\subseteq C$, we have the fact that 
$min\{C'\}$ is no less than  $min\{C\}$ in terms of edge trussness. 
This implies that the ESG upper bound computation can be applied to every subset of $C$ (denoted by $C'$) while the corresponding optimistically estimated c-triangles shall be constructed according to $C'$ and $R$. 
Since each such $C'$ leads to a feasible upper bound, by trying different $C'$, we may have a tighter lower bound. 
Especially in the case where $C'$ contains low trussness edges while weakly connecting to $R$. 
Observed this, we follow \cite{yao2021efficient} and apply edge upper bound for each subset of $C^{\leq\tau}\in C$ s.t. the maximum trussness of edges in $C^{\leq\tau}\in C$ is no greater than each $\tau \in [\tau_{min}(C),\tau_{max}(C)]$. 
The minimum upper bound among these serves as the grouped ESG upper bound.
 
\noindent\textit{\underline{Computational cost}}. It calls the computation of edge upper bound at most $|C|$ times, leading to a total complexity of $O(|C|\cdot \mathrm{ESG()})$, where $\mathrm{ESG()}$ is the complexity of ESG upper bound computation.

\noindent\textbf{Potential of tightening the upper bound}. Since our search framework progressively enlarges $C$ via the depth-first conversion, we pass down the current tightest upper bound and compare it to the edge-grouped upper bound and see if the upper bound can be further refined.


\noindent\textbf{ESG upper bound summary}. Given the fact that all the discussed upper bounds are feasible (including~\cite{zhang2023size}), we always apply the tightest one.

\begin{figure}[t]
  \centering
  \includegraphics[width=0.5\linewidth]{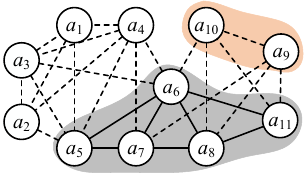}
  \caption{An example of upper-bound}
  \label{fig:expofup}
\end{figure}

\begin{example}
    Refer to Fig.~\ref{fig:expofup}, which is the derived graph $G_P$ from Fig.~\ref{fig:expofhin} when $\mP=(A,P,A),s=7,q=a_6$. Let $C=\{(a_5,a_6),$ $(a_5,a_7)$, $(a_6,a_7),$ $(a_6,a_8),$ $(a_6,a_{11}),$ $(a_7,a_8),$ $(a_8,a_{11})\}$ (highlighted in the gray circle) with the following support values $sup(a_5,a_6) = sup(a_5,a_7) = sup(a_6,a_{11}) = sup(a_7,a_8) = sup(a_8,a_{11})=1$, $sup(a_6,a_7) = sup(a_6,a_8) = 2$ in $G_P[V(C)]$, $R$ contains all edges not connecting to $a_4$ and not in $C$. 
    
    We need to select $s-|V(C)|$=2 nodes with the highest score, namely $\{a_9, a_{10}\}$ (highlighted in the orange circle). Because $a_{10}$ forms 3 triangles with $(a_6, a_8), (a_6, a_{11})$ and $(a_8, a_{11})$, and $a_9$ forms 2 triangles with $(a_7,a_8)$ and $(a_8,a_{11})$. After adding $a_{10}$, the updated support values become $sup(a_5, a_6) = sup(a_5, a_7) = 1$, $sup(a_6, a_7) = sup(a_7, a_8) = sup(a_6,a_8) = sup(a_6,a_{11}) = \\sup(a_8,a_{11})=2$. After further adding $a_9$, the support updates to $sup(a_5, a_6) = sup(a_5, a_7) = sup(a_6, a_7) = sup(a_7, a_8) = sup(a_6,a_8) = sup(a_6,a_{11}) = sup(a_8,a_{11})=2$.\\ Thus, the ESG upper bound is $2+2=4$.

    For the grouped ESG upper bound, firstly, the edges in $C$ can be grouped as follows: $C^{\leq 1} = \{(a_5, a_6), (a_5, a_7), (a_6, a_{11}), (a_7, a_8)$, $(a_8, a_{11})\}$, $C^{\leq 2} = C$. For $C^{\leq 1}$ and $C^{\leq 2}$, the ESG upper bound is $3$ and $4$, respectively. So, the grouped ESG upper bound is $\min\{3,4\}=3$.
    \label{exp:u1}
\end{example}

\subsubsection{Upper bound for NSG}
Unfortunately, our proposed ESG upper bound does not work for NSG. 
This is because these upper bounds are estimated by increasing the minimum edge trussness. In ESG, these bounds work since once an edge is included in $C$, it cannot be removed from $C$. 
This is untrue when $C$ is a set of nodes because edges with low trussness could be removed, and only nodes added to $C$ cannot be removed. 
Therefore, we shall devise upper bounds by exploring other properties.

\noindent\textbf{Rationale}. Reminder that according to Definition~\ref{def:vtruss}, the trussness of a node $u$ for $C$ is the maximum edge trussness for edges incident to $u$ in $G_{P}[C]$. As such, to estimate an upper bound for NSG, we shall assume that triangles are added to edges incident to the node with the maximum trussness, which is totally different from edge-based upper bounds. 

Besides, to secure the optimistic estimation,  we shall consume the least number of triangles so that the trussness of a node can be increased. Luckily, one of the properties of the truss model provides us with such an opportunity, as shown below.
\begin{prop}
    For a ($k,\mP$)-truss, there must be at least $\frac{k(k-1)}{2}$ edges with $\tau(u,v) \geq k$ in $G_{P}$.
    \label{prop:p1}
\end{prop}
Based on the above property, in order to increase the trussness of a subgraph with trussness of $k$ optimistically, we just need to ensure that there are $\frac{(k+1)k}{2}$ edges that are with trussness of $k+1$ after adding some triangles.

We are ready to show the upper bound for NSG given a non-trivial recursion state with $C$ and $R$, where both $C$ and $R$ contain nodes only. 

\noindent\textbf{NSG upper bound}. The upper bound is based on the accurate $\tau_{\triangle}(C)$ in  $G_{P}[C]$, and how $R$ can maximally increase $\tau_{\triangle}(C)$. 

\noindent\textit{\underline{Capable triangle and edge of NSG}}. A triangle is called capable if it is 1) formed by edges with trussness no less than $\tau_{\triangle}(C)$, and 2) triangle reachable from $C$ via triangles satisfying condition 1). A capable edge is defined similarly to the capable triangle. 
 
\noindent\textit{\underline{Optimistic number of capable triangles}}. 
For each node $u\in R$, let its score be the number of capable triangles that $u$ is involved in $G_{P}[C\cup R]$. 
For the top $s-|C|$ nodes with the highest score in $R$, they are assumed to move to $C$, and the sum of their scores is the optimistic number of triangles that contribute to nodes in $C$, denoted by $|\triangle_{opt}|$. 

\noindent\textit{\underline{Optimistic number of edges}}. Additionally, we also need the top $s-|C|$ nodes in $R$ that can potentially contribute the highest number of capable edges, denoted by $|E_{opt}|$. 
We need $|E_{opt}|$ because we assume that edges incident to the $s-|C|$ nodes exhibit sufficiently high trussness. They may directly contribute to nodes in $C$ with the minimum trussness. 

\noindent\textit{\underline{Unifying optimistic triangles and edges}}. To ensure an optimistic result, we assume that there exists a size $s-|C|$ node set that is optimistic in terms of both capable edges and triangles. 

\noindent\textit{\underline{How triangles and edges are optimistically consumed}}. Based on the above discussion, given $C$ with trussness of $k=\tau_{\triangle}(C)$, it is not in $k+1$ truss since it has less than $\frac{k(k+1)}{2}$ edges with trussness of $k+1$ shown in \textsc{Property}~\ref{prop:p1}. So initially, the upper bound is the same as $\tau_{\triangle}(C)$, i.e., $k$. 
Clearly, after adding the discussed $s-|C|$ nodes, it brings $|E_{opt}|$ edges with sufficiently large trussness. As such, we need at least $\frac{k(k+1)}{2}-|E_{opt}|$ (assuming $|E_{opt}|$ is smaller than $\frac{k(k+1)}{2}$)\footnote{otherwise, we first increase $k$ till $\frac{k(k+1)}{2}>|E_{opt}|$.} edges with trussness of $k+1$ to increase the trussness of $C$ by $1$. To do so, we assume that triangles in $\triangle_{opt}$ are used to increase the trussness of edges in $G_{P}[C\cup R]$ whose trussness is closest to $k+1$. Every time the trussness of an edge is increased, one triangle in $\triangle_{opt}$ is consumed. If all triangles in $\triangle_{opt}$ have been consumed and there are still no  $\frac{k(k+1)}{2}-|E_{opt}|$ edges with trussness of $k+1$, the upper bound does not change. If not all the triangles are consumed, the upper bound is increased by $1$. 
If there are triangles in $\triangle_{opt}$ left, we set $k=k+1$ and repeat the above process till all triangles are consumed. The updated upper bound serves as our NSG upper bound.     

\noindent\textit{\underline{Admissibility}}. The admissibility is clear. First, the condition that we increase the trussness of $C$ is the minimum requirement to form a $k+1$-truss. Then, we assume $s-|C|$ nodes bring the most number of capable triangles and edges to $C$. Third, we assume that all the capable triangles contribute to edges that are most likely to increase the trussness, and all the capable edges have sufficient trussness.

Due to the fact that we increase the trussness of edges with the maximum trussness, $\tau_{max}(C')\geq\tau_{max}(C)$ cannot hold anymore for $C'\subseteq C$. As such, we cannot further tighten the upper bound as the one for ESG. 

\noindent\textit{\underline{Computational cost}}. The time complexity of NSG upper bound is also dominated by triangles in $G_{P}[C\cup R]$, i.e., $\mathcal{O}(|E(G_{P}[C\cup R])|^{1.5})$.

\noindent\textbf{NSG upper bound summary}. We can apply the minimum of the upper bound used in~\cite{zhang2023size} and our proposed upper bound to the algorithm. If the derived upper bound is no greater than $\tau_{\triangle}(C^{*})$ (the lower bound up to the time), Algorithm~\ref{alg:nf} can backtrack to other recursion states. 




\begin{example}
    Using Fig.~\ref{fig:expofup} as an example, the query setting and recursion state are the same as in Example~\ref{exp:u1}, but change the edges set to the nodes set. The trussness of all nodes in $C$ is $3$, the two nodes with the highest scores are $a_{10}$ and $a_9$, i.e., $|\triangle_{opt}|=5$ and $|E_{opt}|=7$. To increase the trussness of $C$ to 4, we need $\frac{4 \times 3}{2} - 3 = 3$ edges with trussness of at least 4, which consumes 3 triangles. After this step, the trussness of the edges become: $4$, $4$, $4$, $3$, $3$, $3$, $3$. In the next iteration, to increase the trussness to 5, it would require $\frac{5 \times 4}{2} - (3 + 3 + 1) = 10 - 7 = 3$ edges with trussness of at least 5. Because there are 3 edges with trussness of 4 after the previous iteration, it still requires 3 triangles. However, only 2 triangles remain, which is insufficient to support this increase. Therefore, the NSG upper bound is 4.

\end{example}

\subsection{Branching strategy}\label{sec:branch}
We present a novel branching strategy in this subsection. 

\noindent\textbf{Rationale and SOTA}. Given a recursion state with $C$ and $R$, and two nodes $u$ and $v$ in $R$, if the cohesiveness of every result led by branching at $v$ with $R\setminus \{u\}$ is no more than that of the result by replacing $u$ to $v$, then we can avoid branching at $v$ in the current recursion state while preserving the optimum results. 

For both core and truss-based models, we can examine whether $N(v)\subseteq N(u)$ holds in $G[C\cup R]$ to see if $v$ can be avoided or not. This is referred to $u$ dominating $v$. In this paper, we call such dominance as $u$ edge-dominating $v$ in $G[C\cup R]$, denoted by $u\succeq_{E} v$.    

Based on the above discussion, within their framework, when branching at a node $u$, every node $v$ dominated by $u$ will not be branched, and branches including each pair of $\{u,v\}$ will be attempted, in addition to the normal inclusion and exclusion branches.  

We advance existing ideas in twofold. First, we propose a triangle-based dominance approach, which can be used to expand the node set that can be dominated by a node. Second, we propose a novel branching strategy that aims to minimise the number of branches as much as possible for each recursion state. 

\noindent\textbf{Triangle based dominance}. Given a recursion state $C$ and $R$, let $\triangle(u)$ and $\triangle(v)$ denote the set of triangles that $u$ and $v$ involve in $G[C\cup R]$, if $\triangle(v)\subseteq \triangle(u)$, $u$ is triangle-dominating $v$ in $G[C\cup R]$ denoted by $u\succeq_{\triangle} v$. 

The above discussion is also correct under the context of $G_{P}$. 

\noindent\textbf{Combining edge and triangle dominance}. It is trivial to see that triangle dominance has the same effect as edge dominance. However, since the conditions of triangle dominance are stricter, we must apply it carefully. For a node $u\in R$, we first generate a node set that can be dominated by $u$ using edge-dominance, denoted by $\Phi_{E}(u)$. Then, for nodes that cannot be edge-dominated by $u$ but have common neighbours with $u$, we apply triangle-dominance to generate a node set that can be dominated by $u$, denoted by $\Phi_{\triangle}(u)$. Then $\Phi(u)=\Phi_{E}(u)\cup \Phi_{\triangle}(u)$ is set of nodes that can be dominated by $u$.  

\noindent\textit{\underline{Effectiveness}}. It is easy to see $\Phi(u)\supseteq \Phi_{E}(u)$ always holds. This indicates that, using our proposed dominance, each node can potentially branch at fewer nodes in $R$. 

Now, we are ready to propose the novel branch strategy. 
 \begin{ruledef}
     \label{rule:bran}
     Given a recursion state with $C$ and $R$, select a node with the maximum dominance in $R$, i.e., $u$ is $\argmax_{v}\{\Phi(v)|v\in R\}$. Then, for this recursion state, we only branch at each node $u'$ in $R\setminus \Phi(u)$,i.e., each induced subproblem has an input of $C \cup \{u'\}$ and $R$. 
 \end{ruledef}

\textsc{Rule}~\ref{rule:bran} is a branching rule since it does not remove the dominated nodes from $R$ for the subproblems induced by the current recursion state. Its correctness is clear based on the discussed dominance.  

\noindent\textit{\underline{Computational cost}}. It is evident that the time complexity of computing triangle-based dominance is greater than that of edge-based dominance. Identifying all nodes involved in triangles requires $O(|E(G_P[C \cup R])|^{1.5})$ time, which dominates the overall cost.

\begin{example}
    In Fig.~\ref{fig:expofGp}, $a_3$ dominates $\{a_1, a_2\}$, $a_4$ dominates $\{a_1, a_2, a_3, a_5\}$, and $a_{10}$ dominates $a_{11}$. Therefore, $a_4$ is chosen to reduce the number of branches. Since all these dominances are triangle-based, this example does not clearly highlight the effectiveness of triangle-based dominance over edge-based dominance. But we can consider a simple modification: change the edge $(a_9, a_7)$ to $(a_9, a_5)$. After this change, $a_{10}$ cannot dominate $a_9$ based on edges. However, $a_9$ can still be triangle-dominated, because $(a_9, a_5)$ does not participate in any triangle. As a result, $\triangle(a_9) \subseteq \triangle(a_{10})$, and $a_9$ can be triangle-dominated by $a_{10}$. This can expand the scope of domination of $a_{10}$, illustrating the effectiveness of triangle-based dominance.
\end{example}

Triangle-based dominance can be applied to ESG. Due to the similarity, we do not repeat it.   

\subsection{A total search order and early termination}\label{sec:order}
Thanks to the new framework, we can enforce a total search order for our algorithm, which can improve time complexity as well as bring early termination opportunities. This idea is popular and critical for many other problems, such as maximal clique and biclique enumeration. We now introduce it to the size-bounded community search problem.

\noindent\textbf{Our proposed ordering $\mathcal{O}$}. We propose a distance-first, truss-second-order. I.e., for $R$ that cannot be pruned, we first sort nodes in $R$ in non-decreasing order in terms of their distance to the query node $q$. Then, for the nodes with the same distance, nodes are then sorted in non-increasing order in terms of their trussness in $G_{P}[C \cup R]$.

It is trivial to see that using $\mathcal{O}$ as well as our introduced new framework, Algorithm~\ref{alg:nf} can terminate much earlier if one of the following conditions is satisfied in order. 

\begin{stopC}
For the current $v$, if  $dist(v,q)>1$, Algorithm~\ref{alg:nf} can terminate.  
\end{stopC}

\begin{stopC}
For the current $v$, if $\tau_{\triangle}(v) $ in $G_{P}[C \cup R]$ is no greater than $\tau_{\triangle}(C^{*})$, Algorithm~\ref{alg:nf} can terminate. 
\end{stopC}

\begin{stopC}
If the condition in our bounding technique is satisfied, Algorithm~\ref{alg:nf} can terminate.   
\end{stopC}

The correctness of the above conditions is clear. Later, when analysing the time complexity, \textsc{Stop condition} 1 is critical for reducing the time complexity. 

\begin{example}
    Using Fig.~\ref{fig:expofGp} as an example, let $q = a_6$ and $\tau_{\triangle}(C^*) = 4$. When $R = \{a_2, a_9\}$, Algorithm~\ref{alg:nf} terminates successfully, i.e. \textsc{Stop Condition~1}, because $C\cup R$ is no longer connected to the query node $q$. Furthermore, when $R = \{a_7, a_8, a_9, a_{10}, a_{11}\}$, Algorithm~\ref{alg:nf} triggers \textsc{Stop Condition~2} since there are no nodes with trussness of at least 5. Since the example of \textsc{Stop Condition~3} closely resembles the proposed bounding techniques, we omit it here for brevity.
\end{example}

The above conditions can be refined slightly for ESG. 

\subsection{Candidate reduction}\label{sec:reduction}
In this subsection, we discuss several reduction rules that can be applied to a candidate set $R$. Notice that this subsection is mainly for self-completeness purposes since, except for the reduction rule involving our proposed lower and upper bounds, the distance-based reduction rules are based on properties of trussness and have been used in~\cite{xuEfficientTriangleConnectedTruss2022}.


\begin{ruledef}
    Given a search state $(C,R)$, for any node $v\in R$, if $\tau_{C}(v) \leq \tau_{\triangle}(C^*)$ holds in $G_{P}[C\cup R]$, $v$ can be removed from $R$.
    \label{rule:tauv}
\end{ruledef}

\begin{ruledef}
    Given a search state $(C, R)$, for any node $v \in R$, if $v$ is not $(\tau_{\triangle}(C^*)+1)$-triangle reachable from $C$ in $G_{P}[C\cup R]$, $v$ can be removed from $R$.
    \label{rule:TCv}
\end{ruledef}

\begin{ruledef}
    Given a search state $(C,R)$ with size constraint $s$, for any node $v$ in $R$, define $\varepsilon=\max_{u\in C}\operatorname{dist}_{G_P[C\cup R]}(u,v)$. If $\lceil \frac{\varepsilon(\tau_{\triangle}(C^*)+2)}{2}\rceil> s$ then node $v$ can be removed from $R$.
    \label{rule:disv}
\end{ruledef}

All the above rules can be trivially refined for ESG; therefore, we omit the details.

\begin{figure}[t]
  \centering
  \includegraphics[width=0.5\linewidth]{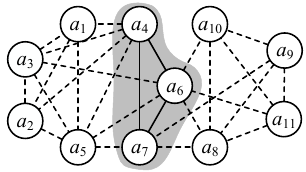}
  \caption{An example of candidate reduction}
  \label{fig:expofedgereduce}
\end{figure}

\begin{example}
    Using Fig. \ref{fig:expofedgereduce} as an example. Let $s=4$, $C=\{a_4,a_6,a_7\}$ (highlighted in the gray circle), $R=\{a_1,a_2,a_3,a_5,a_8,a_9,a_{10},a_{11}\}$ and $\tau_{\triangle}(C^*)=3$.
    
    By Rule~\ref{rule:tauv}, no nodes can be pruned, because all nodes have trussness of at least $4$. Additionally, by Rule~\ref{rule:TCv}, nodes $\{a_8,a_9,a_{10},a_{11}\}$ are excluded from $R$, because they are not $4$-triangle-connected with $C$. According to Rule~\ref{rule:disv}, for $a_1$, It has $\varepsilon=2$ and $\lceil\frac{2\times(3+2)}{2}\rceil=5>4$ thus, $a_1$ can be discarded from $R$. 
    The same applies to $a_2$ and $a_3$.
    As such, the updated state becomes $C = \{a_4, a_6, a_7\}$ and $R = \{a_5\}$.
\end{example}

\subsection{Wrap-up}\label{sec:wrap}
In this subsection, we wrap up with the time complexity analysis.

\noindent\textbf{NSG}. In terms of $G_{P}$, the time complexity is $ \mathcal{O}(|N(q,G_{P})|$ $\binom{|V(G)|}{s-2}$ $|E(G_{P})|^{1.5})$. Computing $G_{P}$ can be bounded by $\mathcal{O}(|V(G)||E(G)|)$, where $G$ is the input HIN. 
$G_{P}$ can be computed much fast according to the query meta-path. 
The main cost of each recursion is induced by computing triangles in $G_{P}[C\cup R]$, which can be bounded by $|E(G_{P})|^{1.5}$ and should be much faster since we do not compute it from scratch. Compared to the complexity mentioned in Section~\ref{sec:frame}, the exponential part is reduced to  $|N(q, G_{P})|\binom{|V(G)|}{s-2}$ thanks to our proposed total ordering and early \textsc{Stop condition 1}, i.e., line 4 in Algorithm~\ref{alg:heu} is bounded by $|N(q,G_{P})|$ and Algorithm~\ref{alg:heu} only needs to check up to $\binom{|V(G)|}{s-2}$ size node sets.    

\noindent\textbf{ESG}. The time complexity is $\mathcal{O}$ $(|N(q,G_{P})|$ $\binom{|E(G)|}{\frac{(s-1)(s-2)}{2}}$ $|E(G_{P})|^{1.5})$. It has a much higher exponential part compared to NSG. This is because we have to add $\frac{(s-1)(s-2)}{2}$ edges so that the node size can reach $s-1$ in the worst case. 

Although we can explore more pruning opportunities for ESG, NSG exhibits much lower time complexity. We will conduct comprehensive experiments to evaluate their practical performance.

\section{Experiments}
\label{sec:exp}
In this section, we present experimental results to demonstrate the effectiveness and efficiency of the proposed methods.

\noindent\textbf{Datasets.}
We use $5$ real datasets.
Their statistics are shown in Table \ref{datasets}, where $|E(G_P)|$ represents the average edge count of the graphs $G_P$ derived from different meta-paths $\mP$ in the meta-path pool.
The first four datasets are domain-specific.
Specifically, Amazon is a platform from the e-commerce domain, encompassing users, products, reviews, categories, and brands. 
DoubanMovie is a movie-related dataset containing entities such as movies, genres, users, groups, directors, and actors. 
DBLP is a bibliographic dataset comprising authors, papers, topics, and venues. 
Aminer is also a bibliographic dataset but includes additional node types beyond those in DBLP, such as publication years and references.
The last dataset, Freebase, is a large-scale knowledge graph encompassing diverse domains, including music, film, and sports.

\begin{table}[h]
\caption{Datasets}
\label{datasets}
\small
\centering
\setlength{\tabcolsep}{4pt}
\begin{tabular}{ccccccc}
\hline
Datasets    & Nodes      & Edges      & $|E(G_P)|$    & $|\mathcal{A}|$ & $|\mathcal{R}|$ & \#$\mP$ \\ \hline
Amazon      & 13,136     & 209,746    & 5,929,707     & 5               & 4               & 10                      \\
DBLP        & 37,791     & 170,794    & 4,501,017     & 4               & 4               & 10                      \\
DoubanMovie & 37,595     & 1,713,902  & 26,234,499    & 6               & 6               & 10                      \\
Aminer      & 439,504    & 1,130,039  & 354,466,972   & 6               & 6               & 10                      \\
Freebase    & 12,116,680 & 35,231,642 & 1,428,764,912 & 8               & 36              & 100                     \\ \hline
\end{tabular}
\end{table}

\noindent
\textbf{Algorithms.} We compare the following algorithms.
\begin{itemize}
    \item SC-BRB \cite{yao2021efficient}: size constraint $k^{*}$-core.
    \item ST-Exa \cite{zhang2023size}: size constraint $k^{*}$-truss.
    \item SCSHEP$^+$: ESG equipped with all techniques.
    \item SCSHEV$^+$: NSG equipped with all techniques.
\end{itemize}

Due to the fact that SC-BRB and ST-Exa are designed for homogeneous graphs, we adapt them with $G_P$ as input.

\noindent
\textbf{Queries.} For each dataset, we construct a meta-path pool, with sizes reported in Table \ref{datasets}. The maximum meta-path length is 6, following established settings in prior studies \cite{fang2020effective} and \cite{yang2020effective}. The size bound $s$ varies from 9 to 21 in increments of 3, i.e., $s\in\{9,12,15,18,21\}$, consistent with previous works \cite{yao2021efficient, zhang2023size}.

We generate $100$ queries for each dataset. To generate queries, a meta-path is randomly selected from the pool, and $10$ query nodes corresponding to the selected meta-path are sampled. 
For each query node, we evaluate all size bounds from $9$ to $21$. Each query execution is constrained by a one-hour time limit. 
All algorithms are evaluated on the same set of queries for fair comparison and implemented in Python, running on a Windows system with an Intel(R) Xeon(R) W-2133 CPU @ 3.60GHz and 32 GB RAM.

\subsection{Effectiveness Evaluation}
We evaluate the effectiveness of SCSHEV$^+$, SC-BRB, and ST-Exa w.r.t., focusing on two widely used metrics: density and similarity.

\begin{figure*}[h]
    \centering
    \vspace{-10pt}
    \subfigure{
        \includegraphics[width=.2\textwidth]{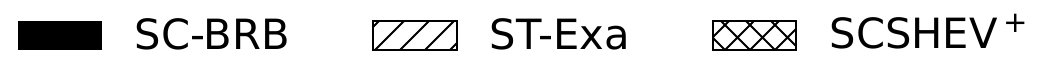}
    }
    \\[-10pt]
    \subfigure[Amazon]{
        \label{fig:Amazondensity}
        \includegraphics[width=.18\textwidth]{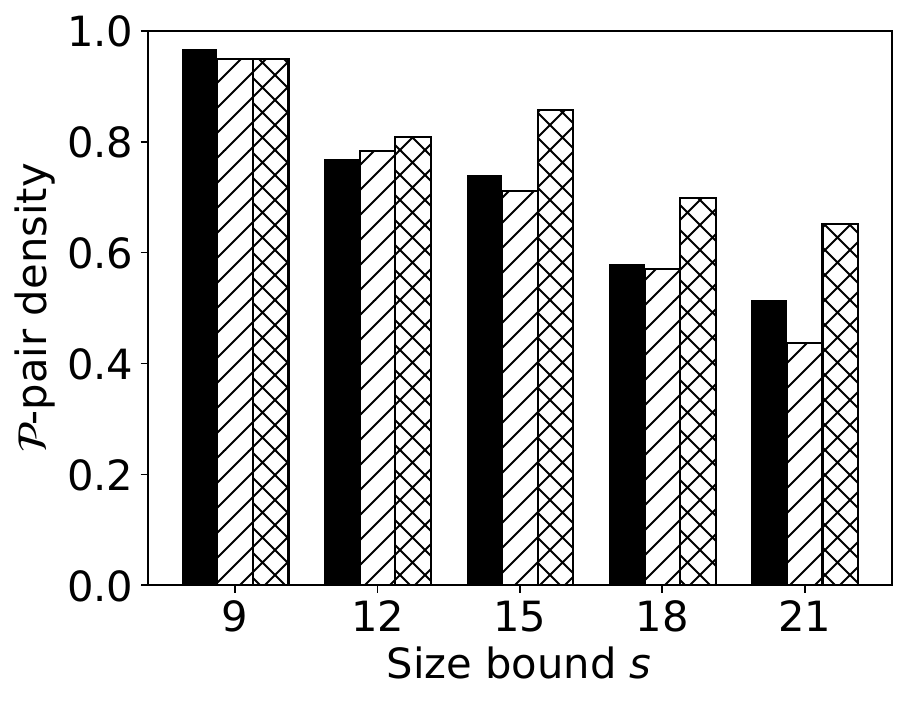}
    }
    \subfigure[DBLP]{
        \label{fig:DBLPdensity}
        \includegraphics[width=.18\textwidth]{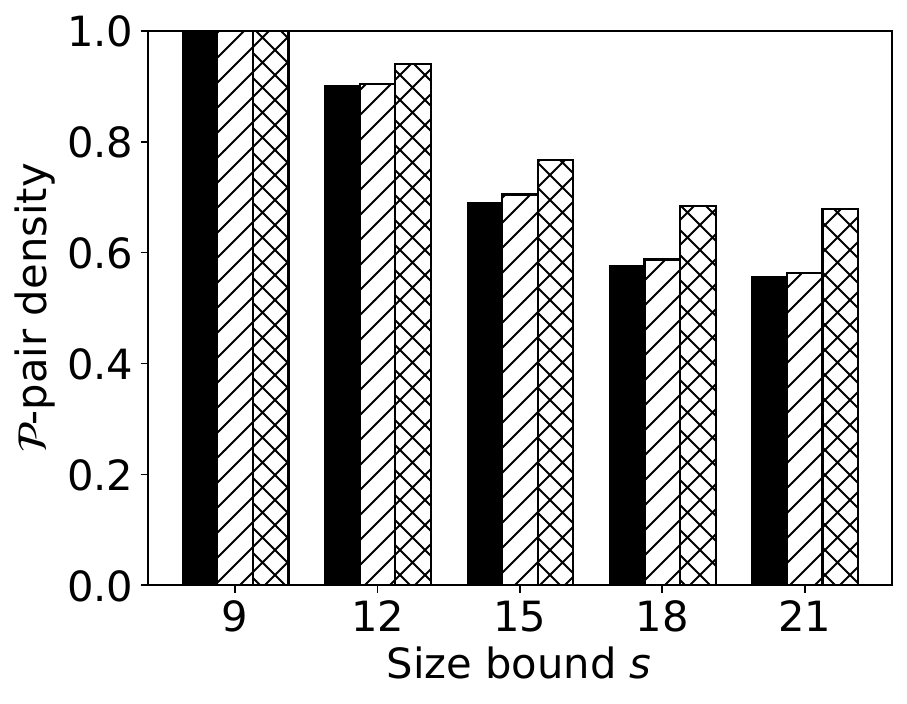}
    }
    \subfigure[DoubanMovie]{
        \label{fig:DoubanMoviedensity}
        \includegraphics[width=.18\textwidth]{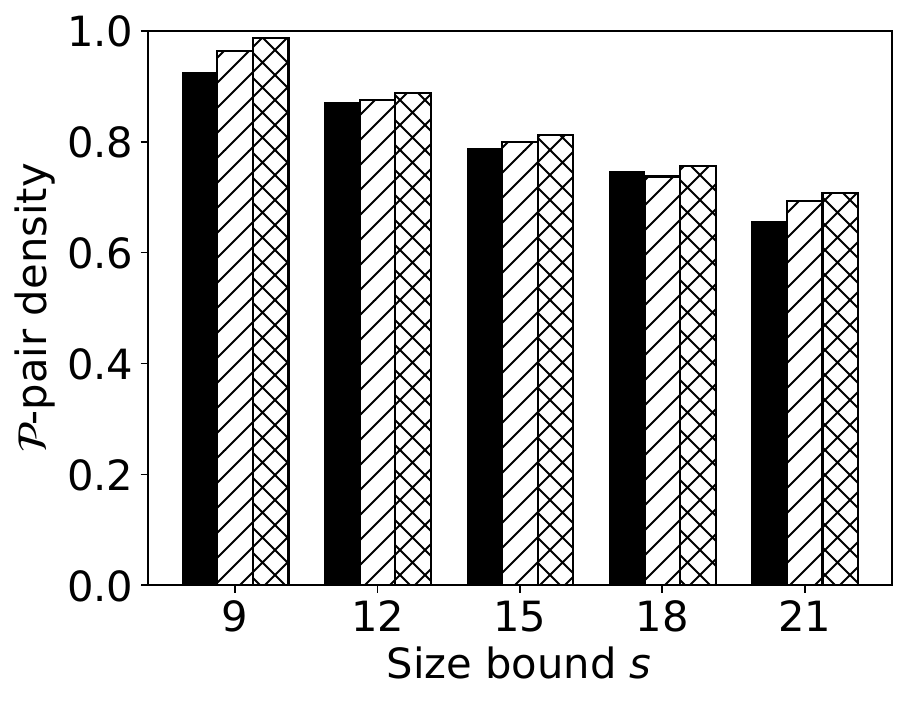}
    }
    \subfigure[Aminer]{
        \label{fig:Aminerdensity}
        \includegraphics[width=.18\textwidth]{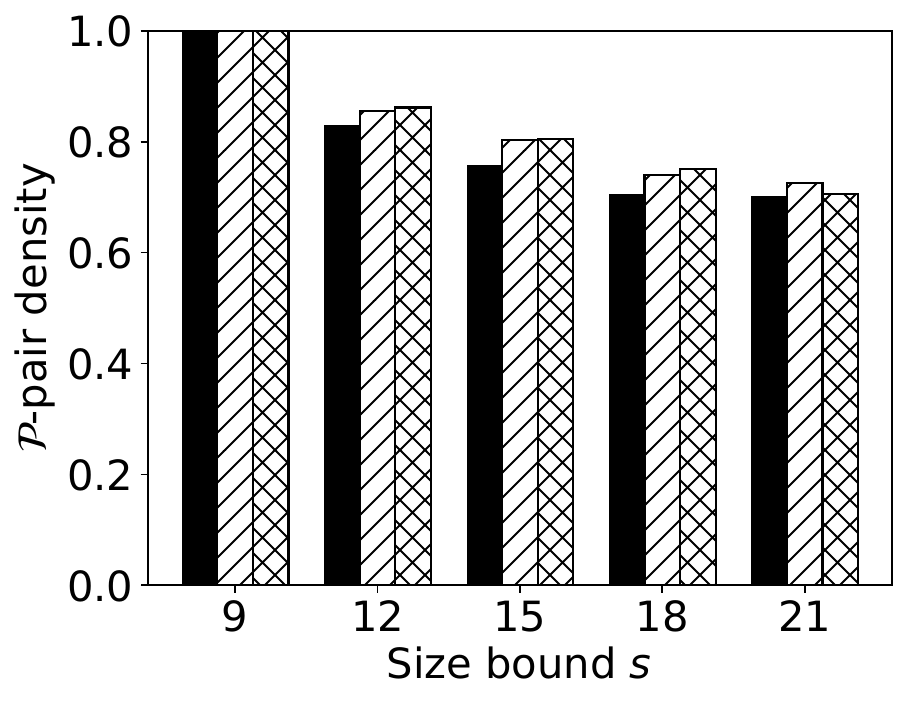}
    }
    \subfigure[Freebase]{
        \label{fig:Freebasedensity}
        \includegraphics[width=.18\textwidth]{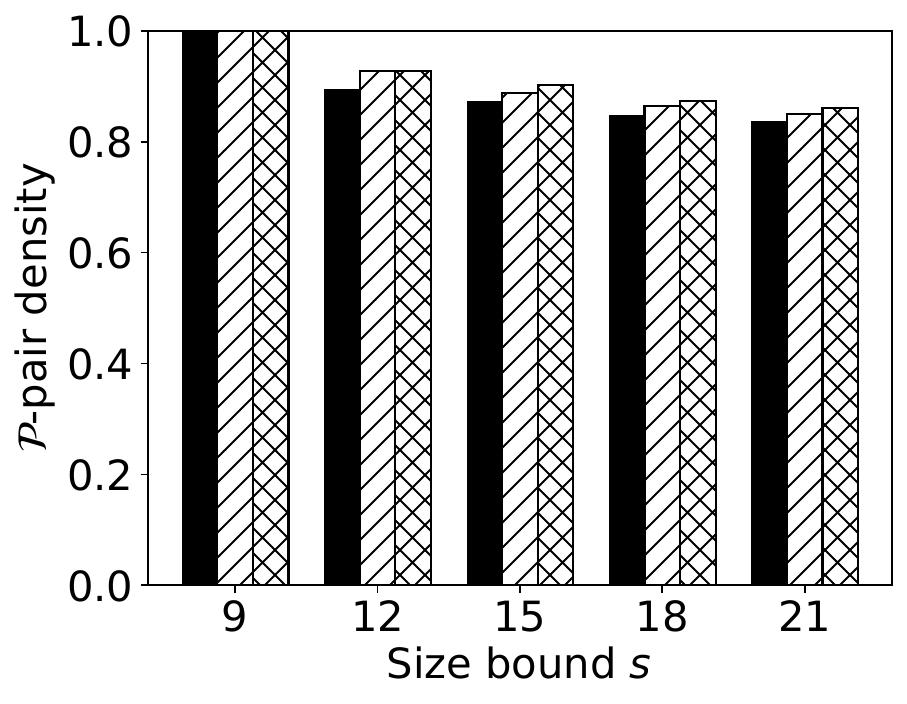}
    }
    \caption{$\mP$-pair density}
    \label{fig:density}
\end{figure*}

\begin{figure*}[h]
    \centering
    \vspace{-10pt}
    \subfigure{
        \includegraphics[width=.2\textwidth]{figures/densitylegend.pdf}
    }
    \\[-10pt]
    \subfigure[Amazon]{
        \label{fig:Amazonsimilarity}
        \includegraphics[width=.18\textwidth]{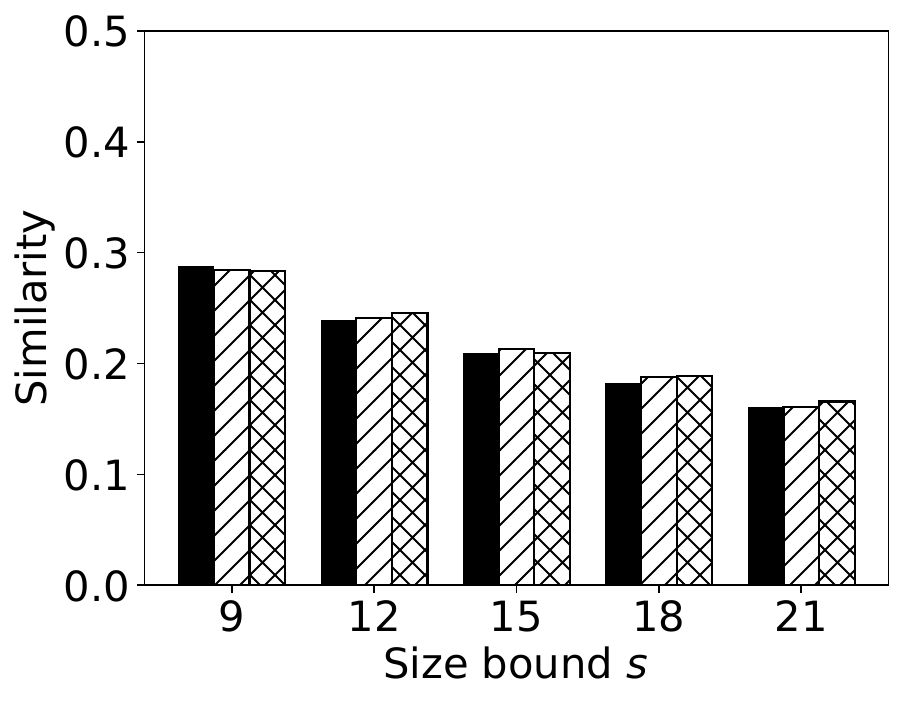}
    }
    \subfigure[DBLP]{
        \label{fig:DBLPsimilarity}
        \includegraphics[width=.18\textwidth]{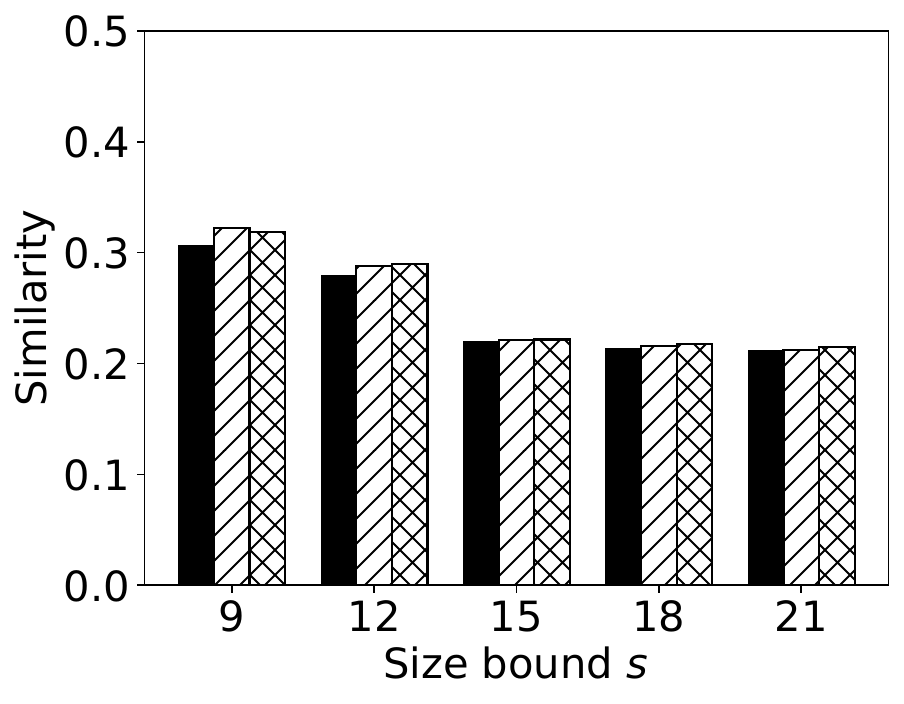}
    }
    \subfigure[DoubanMovie]{
        \label{fig:DoubanMoviesimilarity}
        \includegraphics[width=.18\textwidth]{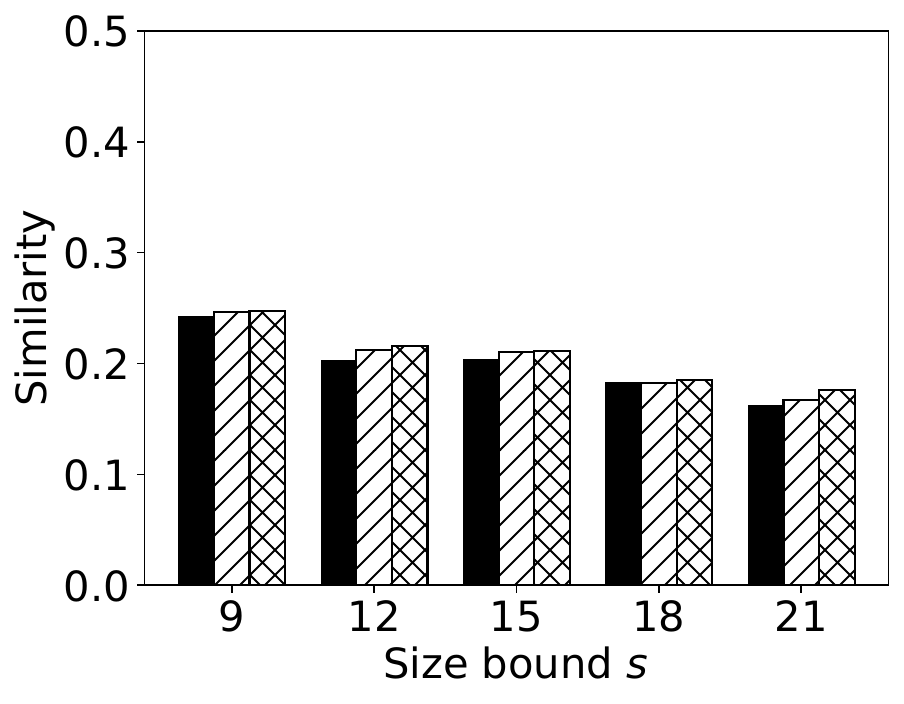}
    }
    \subfigure[Aminer]{
        \label{fig:Aminersimilarity}
        \includegraphics[width=.18\textwidth]{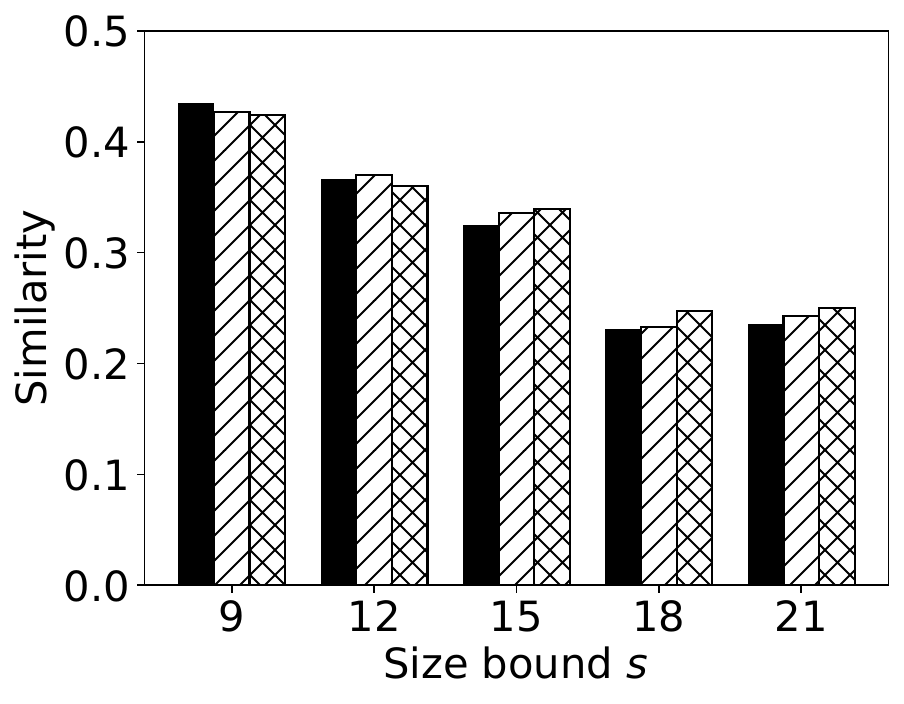}
    }
    \subfigure[Freebase]{
        \label{fig:Freebasesimilarity}
        \includegraphics[width=.18\textwidth]{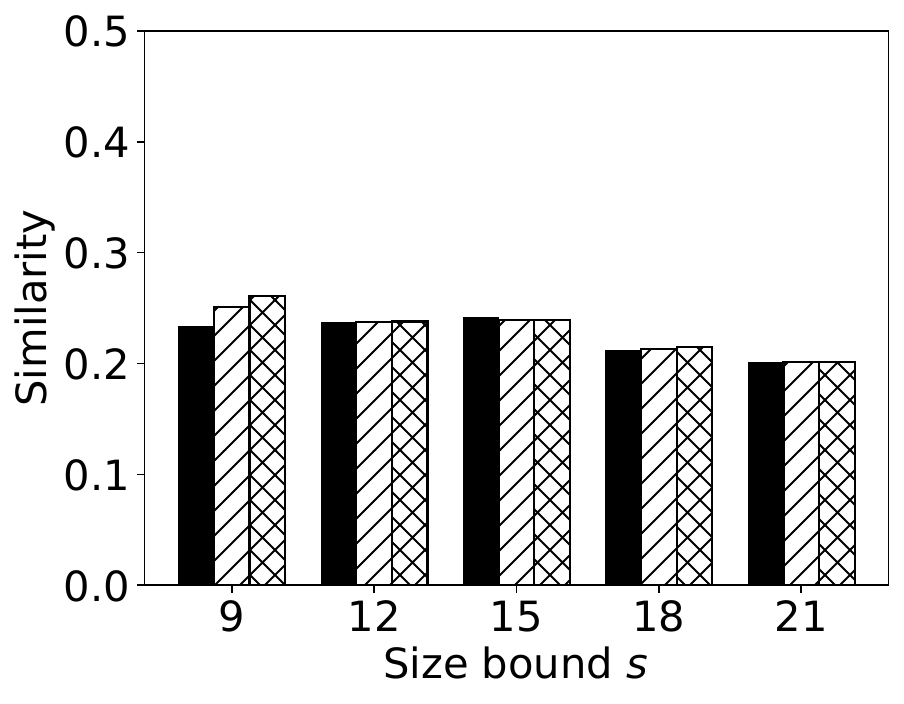}
    }
    \caption{Similarity}
    \label{fig:similarity}
\end{figure*}

\noindent\textbf{Density of $\mP$-pairs}. The concept of $\mP$-pair density, introduced in \cite{fang2020effective}, is the number of $\mP$-pairs over the number of target nodes that comprise the $\mP$-pairs. Fig. \ref{fig:density} presents the density results across all datasets. As the size bound $s$ increases, density generally decreases, which aligns with the intuition that larger subgraphs are less likely to be fully connected. Among the compared methods, SCSHEV$^+$ consistently achieves the highest density, owing to its stricter structural constraints compared to $k$-core and $k$-truss, enabling it to identify denser substructures. Notably, in the DBLP, Aminer, and Freebase datasets, SCSHEV$^+$ achieves a density of 1 when $s=9$, indicating that it discovers cliques under small size bounds.


\noindent\textbf{Similarity}. We use PathSim\cite{sun2011pathsim} to compute node similarity, and define the similarity of a community as the average PathSim score across all node pairs within a community. The results are presented in Fig. \ref{fig:similarity}. As the size bound $s$ increases, the average similarity tends to decrease, which is expected since larger communities are less likely to maintain high internal similarity. In most cases, SCSHEV$^+$ achieves the highest similarity scores because its communities are more densely connected, leading to higher semantic similarity among members.
   
\noindent\textbf{Case study}. In the case study, we use the Amazon dataset with query node i385, the meta-path $\mP = \textit{(item, view, item)}$, and a size bound of 15. Fig.~\ref{fig:SCBRBcase} presents the result returned by SC-BRB, which produces a 6-core with a $\mP$-pair density of 0.476. Fig.~\ref{fig:SCSHEVcase} shows the result returned by SCSHEV$^+$, which yields a 5-core that also forms a triangle-connected 6-truss, with a higher $\mP$-pair density of 0.562.

Now, consider a user searching for product i385, and the platform aims to recommend related products. As shown in Fig.~\ref{fig:SCBRBcase}, the result consists of two loosely connected communities, indicating that the recommended products are not strongly related to i385. In contrast, Fig.~\ref{fig:SCSHEVcase} reveals that most of the nodes belong to the same tightly-knit community as i385, suggesting stronger semantic relevance and a more effective recommendation outcome.

\begin{figure}[h]
    \centering
    \subfigure[SC-BRB]{
        \label{fig:SCBRBcase}
        \includegraphics[width=.28\textwidth]{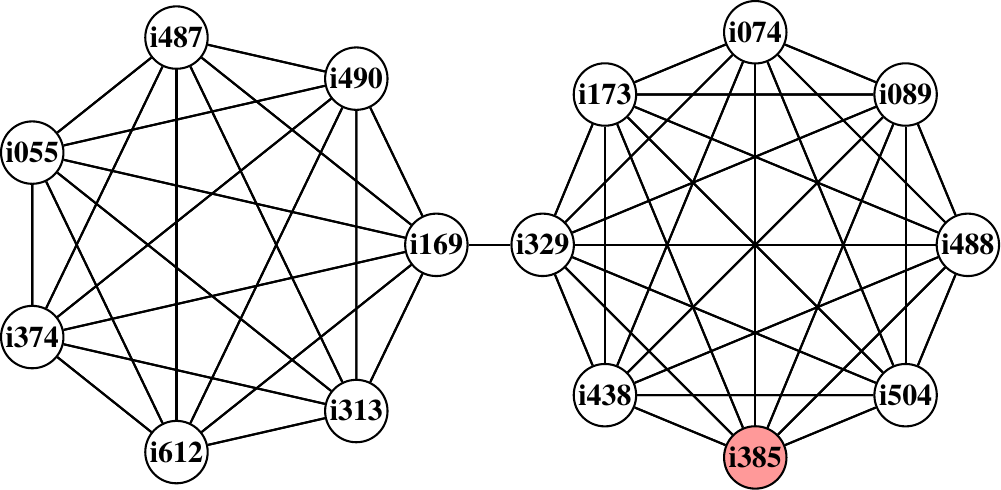}
    }
    \subfigure[SCSHEV$^+$]{
        \label{fig:SCSHEVcase}
        \includegraphics[width=.17\textwidth]{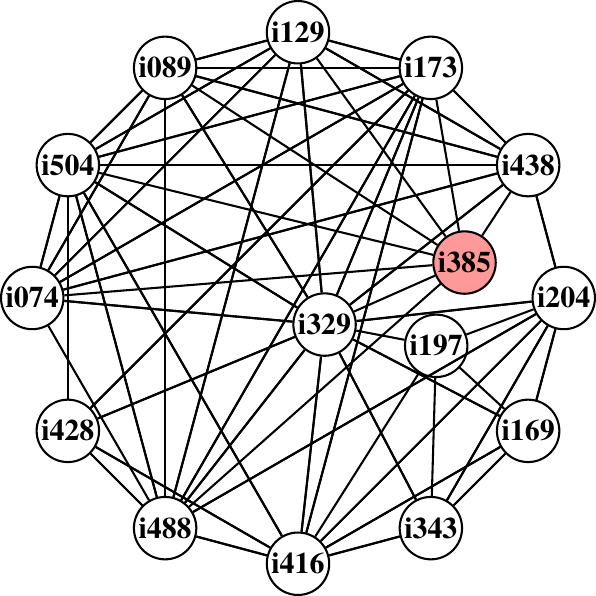}
    }

    \caption{Case study}
    \label{fig:DT}
\end{figure}

\subsection{Efficiency Evaluation}

\begin{figure*}[h]
    \centering
    \vspace{-10pt}
    \subfigure{
        \includegraphics[width=.2\textwidth]{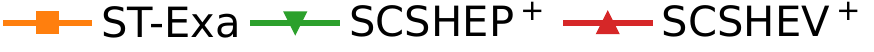}
    }
    \\[-10pt]
    \subfigure[Amazon]{
        \label{fig:Amazonefficiency}
        \includegraphics[width=.18\textwidth]{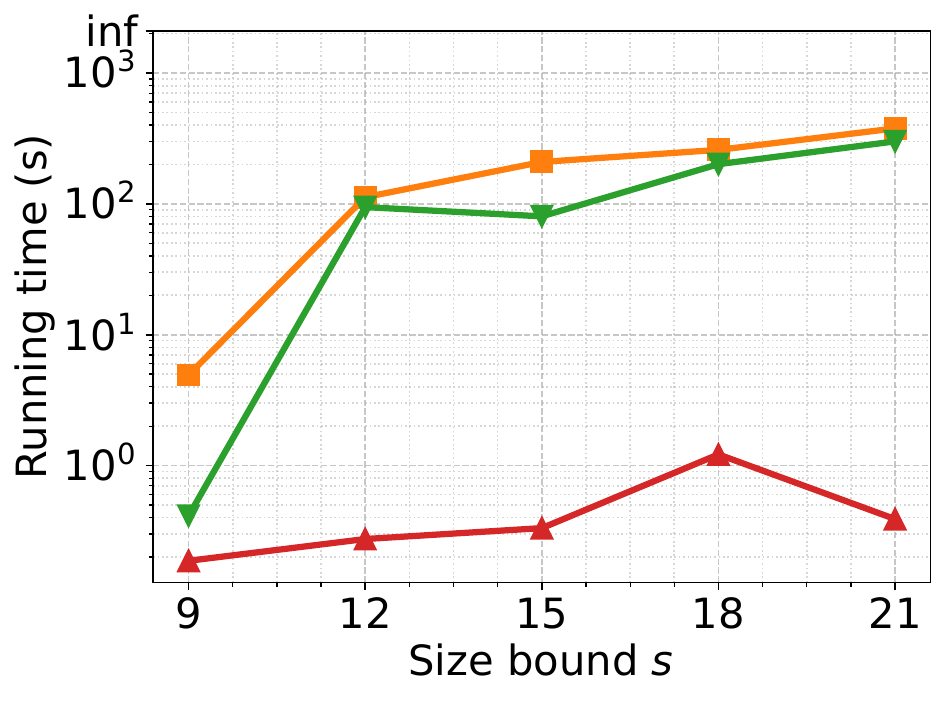}
    }
    \subfigure[DBLP]{
        \label{fig:DBLPefficiency}
        \includegraphics[width=.18\textwidth]{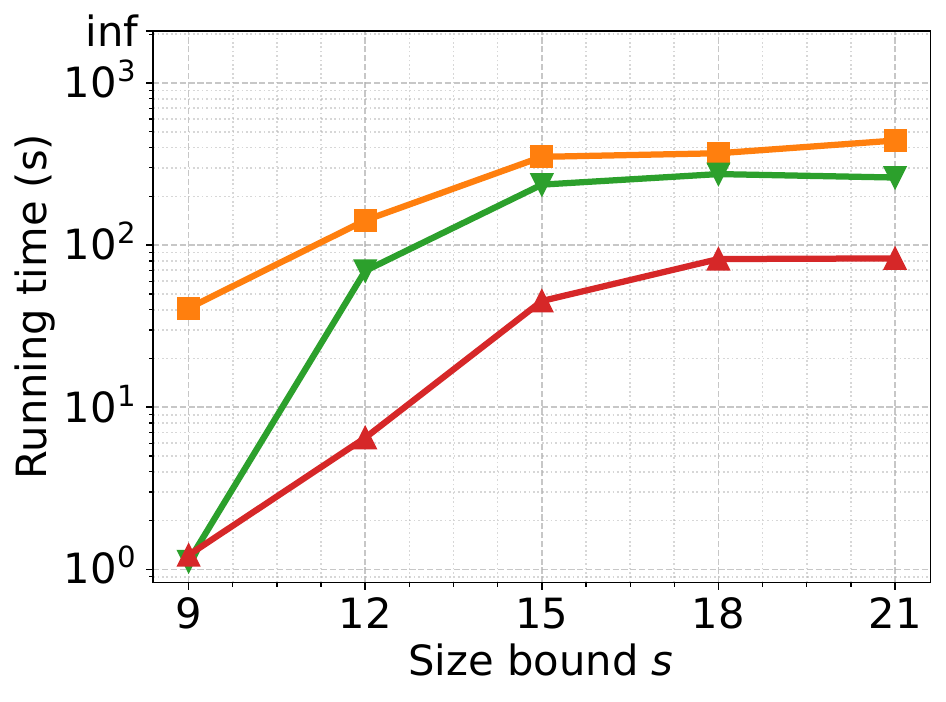}
    }
    \subfigure[DoubanMovie]{
        \label{fig:DoubanMovieefficiency}
        \includegraphics[width=.18\textwidth]{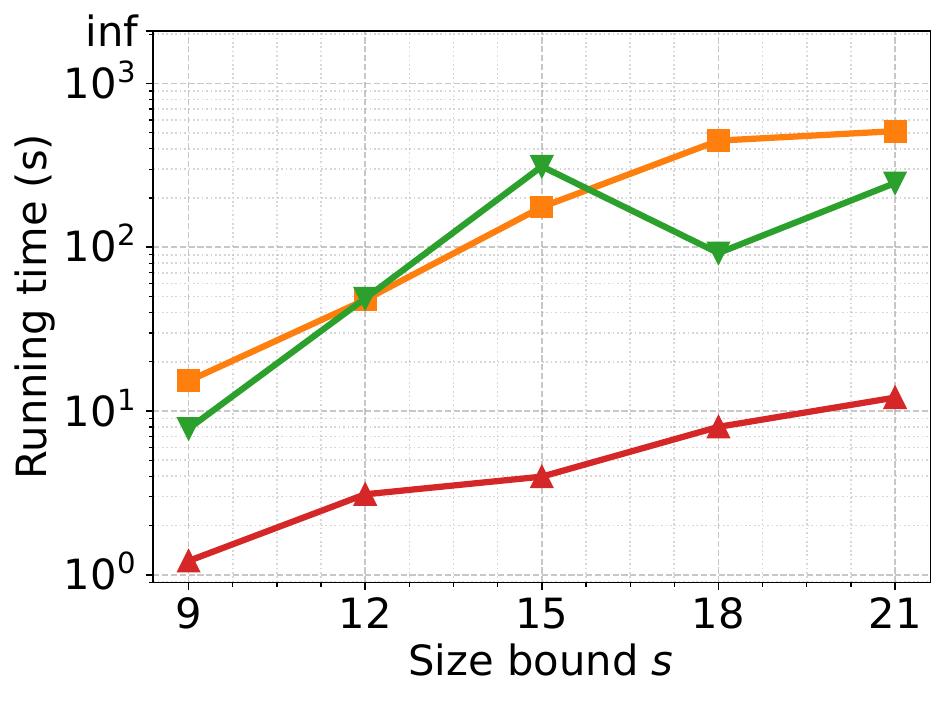}
    }
    \subfigure[Aminer]{
        \label{fig:Aminerefficiency}
        \includegraphics[width=.18\textwidth]{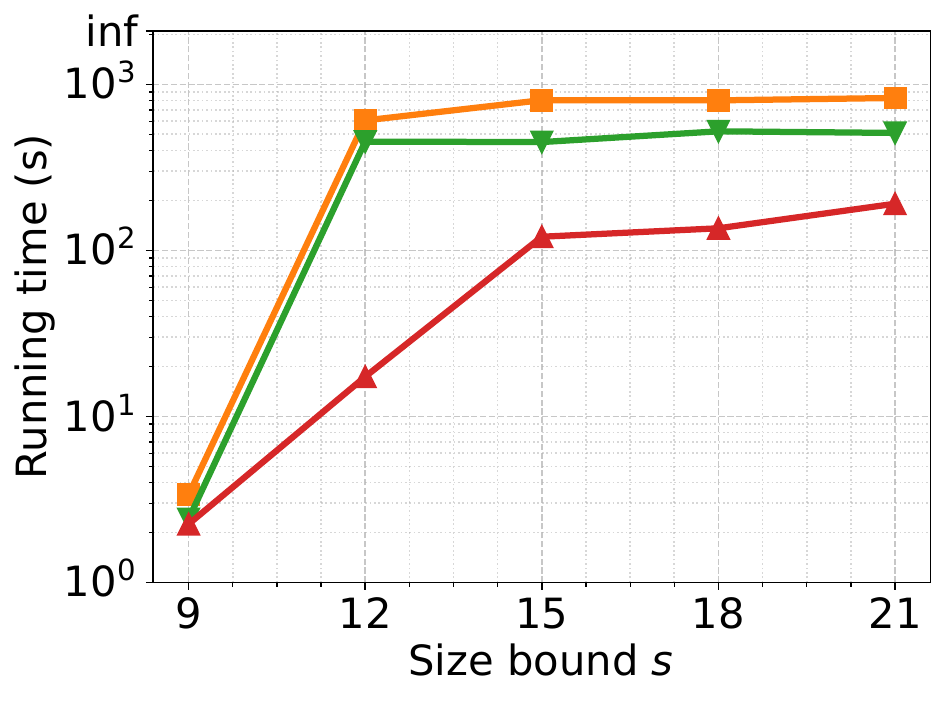}
    }
    \subfigure[Freebase]{
        \label{fig:Freebaseefficiency}
        \includegraphics[width=.18\textwidth]{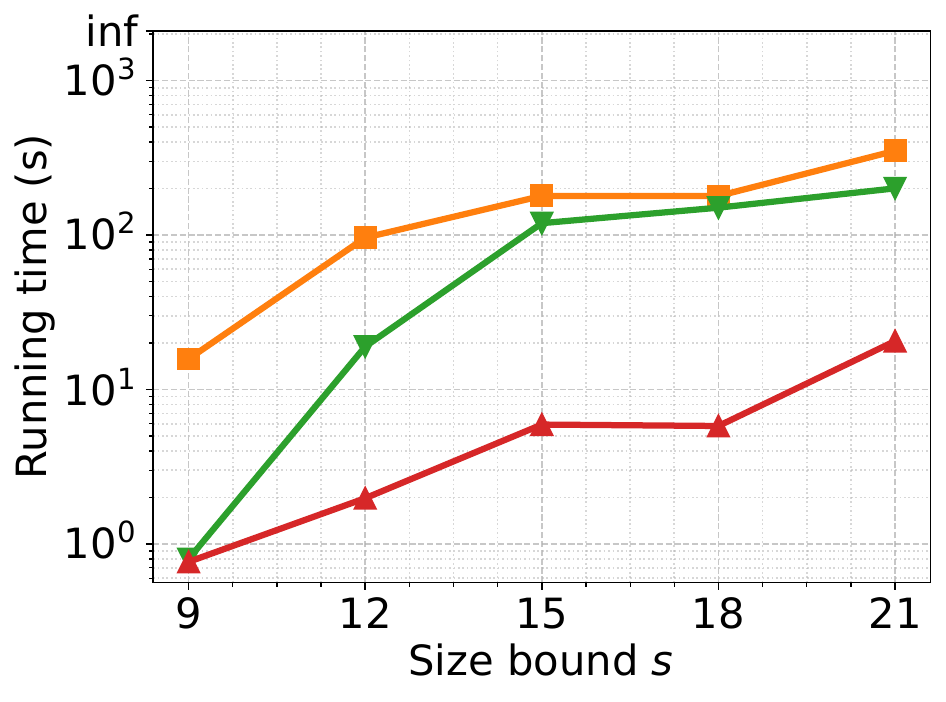}
    }
    \caption{Overall running times}
    \label{fig:efficiency}
\end{figure*}

SC-BRB is not applicable to the SCSH problem. We slightly modify ST-Exa to make it compatible with our setting. Specifically, when $|C| = s$, the computation of trussness is replaced by triangle-connected trussness.


\noindent\textbf{Against baseline algorithms}. The overall running time is presented in Fig.~\ref{fig:efficiency}. As the size bound $s$ increases, the running time also increases. This is because a larger $s$ typically leads to a smaller $k$, which in turn reduces the number of pruned branches during the search.  However, the running time does not always increase monotonically with $s$. For example, on the DoubanMovie dataset, the running time of SCSHEP$^+$ decreases significantly when $s$ increases from 15 to 18 and even becomes worse than ST-Exa when $s=15$. This is because the heuristic is able to identify a better lower bound when $s=18$, which greatly enhances the effectiveness of pruning and reduction. In contrast, when $s=15$, the heuristic fails to find such a bound, resulting in a larger search space and longer running time despite the smaller size bound. On the other hand, the running time of SCSHEV$^+$ does not exhibit a similar drop over the same interval. This is because the gain from the better lower bound at $s=18$ is offset by the increased search space due to the larger size bound.

Among all datasets, SCSHEV$^+$ consistently achieves the best running time. As discussed in Section~\ref{sec:wrap}, although the ESG upper bound could be tighter than the NSG one, SCSHEP$^+$ performs worse than SCSHEV$^+$ in practice due to its high time complexity.

\noindent\textbf{Different techniques}. To evaluate the effectiveness of different techniques, we implement four variants: SCSHEV$^+$/H  heuristic lower bound, SCSHEV$^+$/R without reduction rules, SCSHEV$^+$/P without pruning rules based on upper bound, and SCSHEV$^+$/B without branching rules. Fig.~\ref{fig:DT} shows their running times on the Amazon and Aminer datasets as $s$ varies. We observe that running time increases with $s$, not only due to a larger search space but also because higher $s$ tend to lower the trussness of the final result, reducing the effectiveness of both reduction and pruning. Among all algorithms, SCSHEV$^+$ consistently achieves the best performance. Moreover, removing any of the techniques leads to increased running times, confirming their contributions. Without reduction rules, the impact is the most pronounced, resulting in the largest increase in running time.

\begin{figure}[h]
    \centering
    \vspace{-10pt}
    \subfigure{
        \includegraphics[width=.4\textwidth]{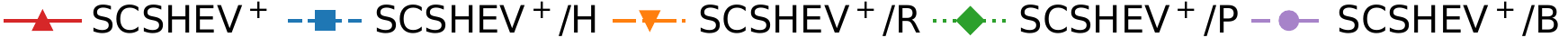}
    }
    \\[-10pt]
    \subfigure[Amazon]{
        \label{fig:AmazonDT}
        \includegraphics[width=.20\textwidth]{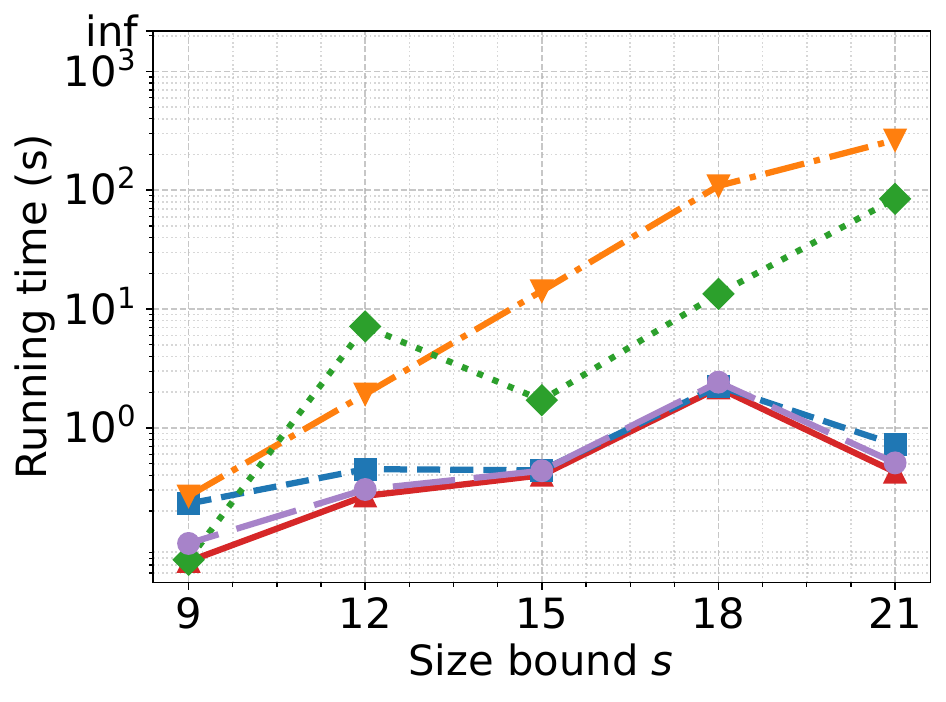}
    }
    \subfigure[Aminer]{
        \label{fig:AminerDT}
        \includegraphics[width=.20\textwidth]{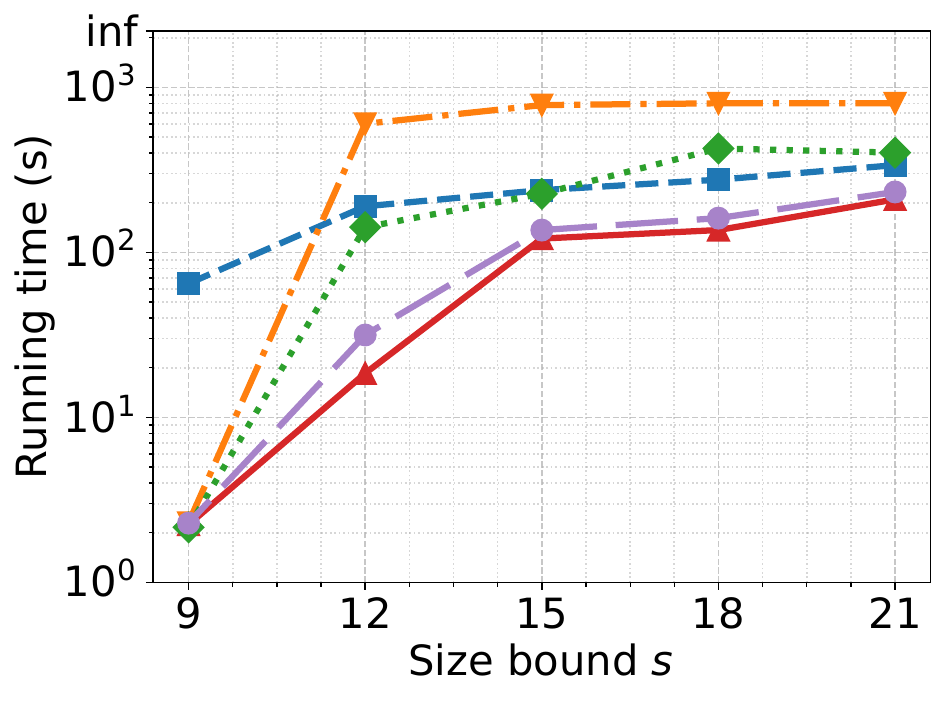}
    }

    \caption{Evaluate different techniques}
    \label{fig:DT}
\end{figure}

\noindent\textbf{Different reduction rules}. We also evaluate the impact of individual reduction rules on running time. Here, $R_1$ - $R_4$ correspond to Rule~\ref{rule:tauv}-Rule~\ref{rule:disv}. We incrementally incorporate these rules into SCSHEV, for example, SCSHEV+$R_2$ includes both $R_1$ and $R_2$, and so on. The running times on the Amazon and Aminer datasets are shown in Fig.~\ref{fig:DR}. From the figure, we observe that the algorithm with all reduction rules, SCSHEV+$R_3$, achieves the best performance. Each rule contributes to performance improvement, though to varying degrees. Some rules have a stronger impact, while others offer more modest gains or have overlapping effects, such as $R_2$ and $R_3$, resulting in smaller improvements.

\begin{figure}[h]
    \centering
    \vspace{-10pt}
    \subfigure{
        \includegraphics[width=.33\textwidth]{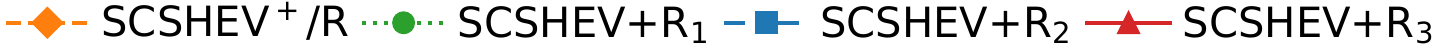}
    }
    \\[-10pt]
    \subfigure[Amazon]{
        \label{fig:AmazonDT}
        \includegraphics[width=.20\textwidth]{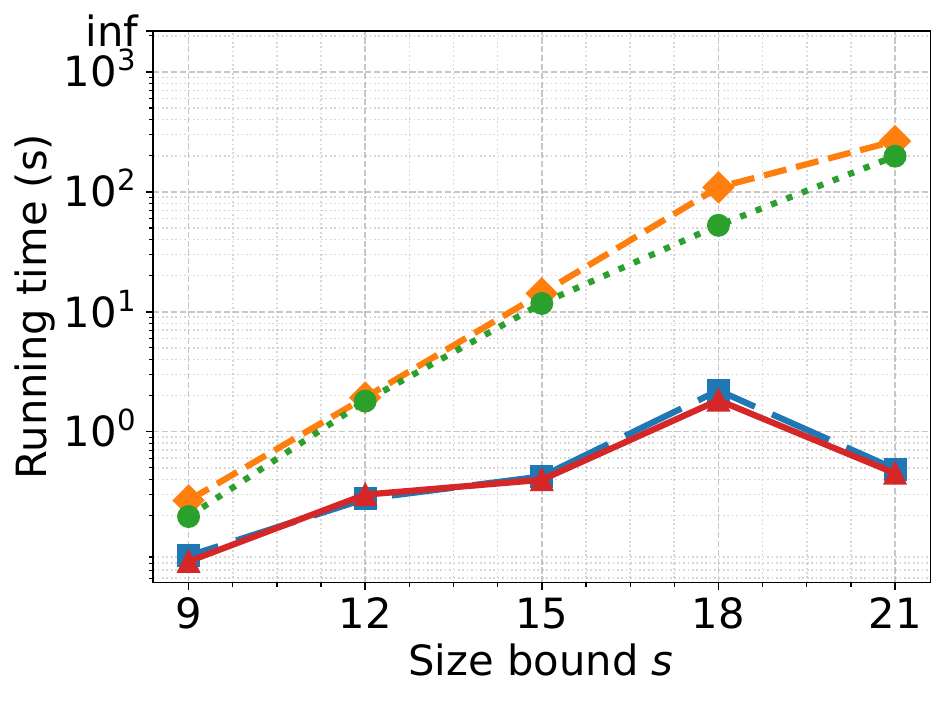}
    }
    \subfigure[Aminer]{
        \label{fig:AminerDT}
        \includegraphics[width=.20\textwidth]{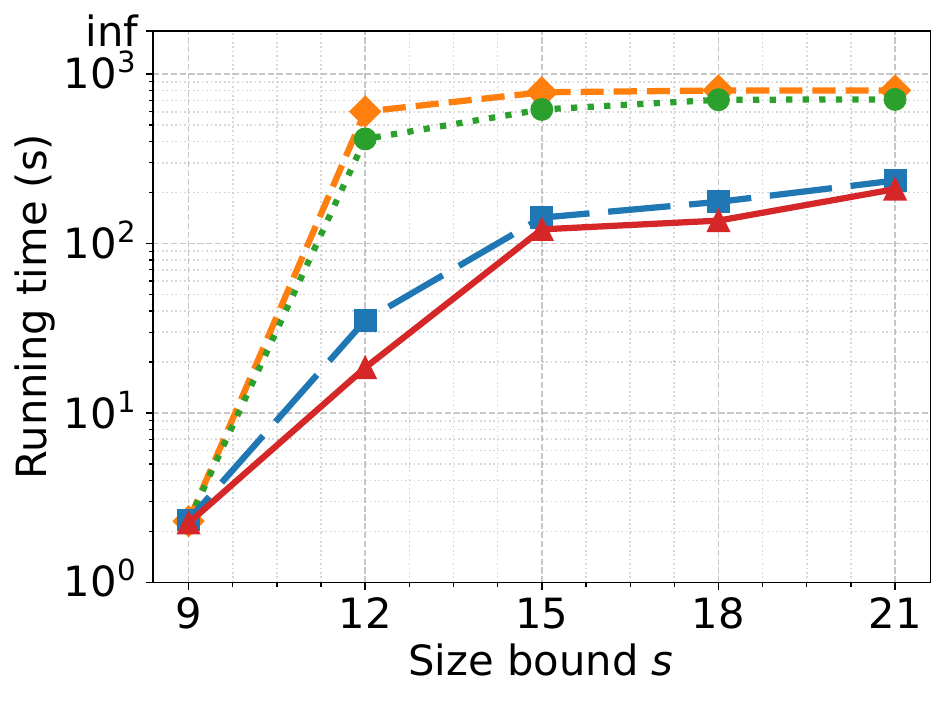}
    }

    \caption{Evaluate different reduction rules}
    \label{fig:DR}
\end{figure}

\noindent\textbf{Impacts of upper bounds}. Fig.~\ref{fig:DU} illustrates the impact of upper bounds. Following the idea from \cite{yao2021efficient}, we define three upper bounds: $U_1 = \min_{v \in C}(\tau_{G_P[C \cup R]}(v))$, $U_2 = \min_{v \in C}(\tau_{G_P[C]}(v)) + s - |C|$, and $U_3$, which is our proposed NSG upper bound. Based on these, we introduce three pruning strategies: $P_1$ corresponds to $UB(C, R) = U_1$, $P_2$ to $\min\{U_1, U_2\}$, and $P_3$ to $\min\{U_1, U_2, U_3\}$. As shown in the figure, SCSHEV+$P_3$ achieves the best performance, demonstrating the effectiveness of the pruning rules. On the Amazon dataset, $P_2$ and $P_3$ yield similar results, while on the Aminer dataset, $P_1$ and $P_2$ perform comparably. This suggests that, although the rules have overlapping effects, they are not redundant. The extent of overlap may vary depending on the underlying topology of the datasets.

\begin{figure}[h]
    \centering
    \vspace{-10pt}
    \subfigure{
        \includegraphics[width=.33\textwidth]{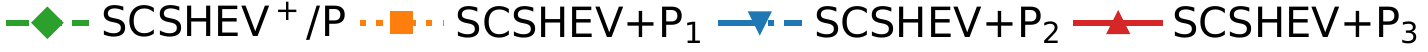}
    }
    \\[-10pt]
    \subfigure[Amazon]{
        \label{fig:AmazonDT}
        \includegraphics[width=.20\textwidth]{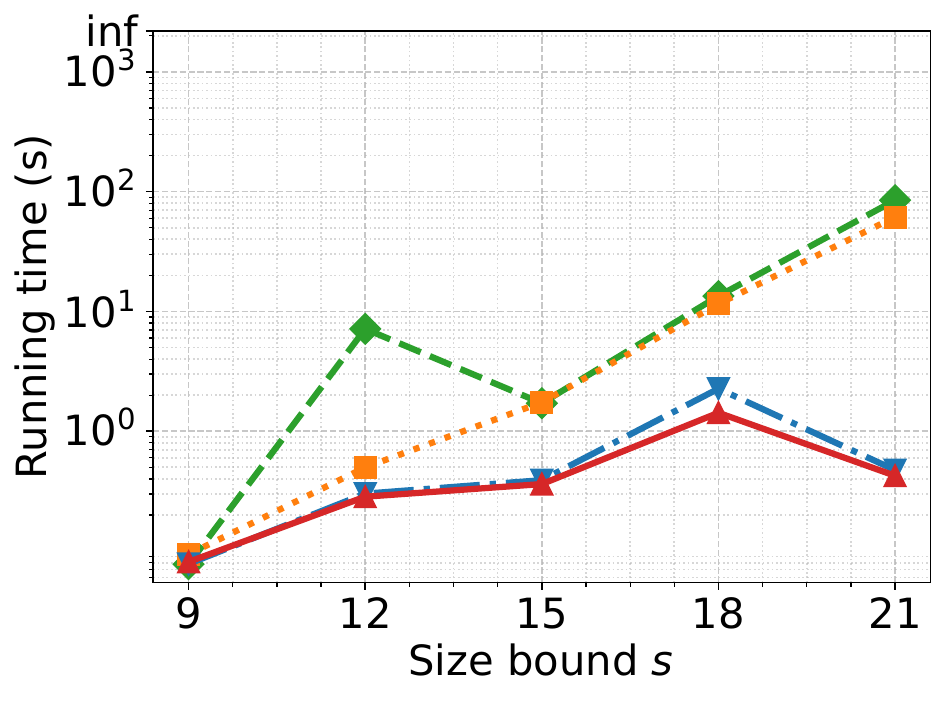}
    }
    \subfigure[Aminer]{
        \label{fig:AminerDT}
        \includegraphics[width=.20\textwidth]{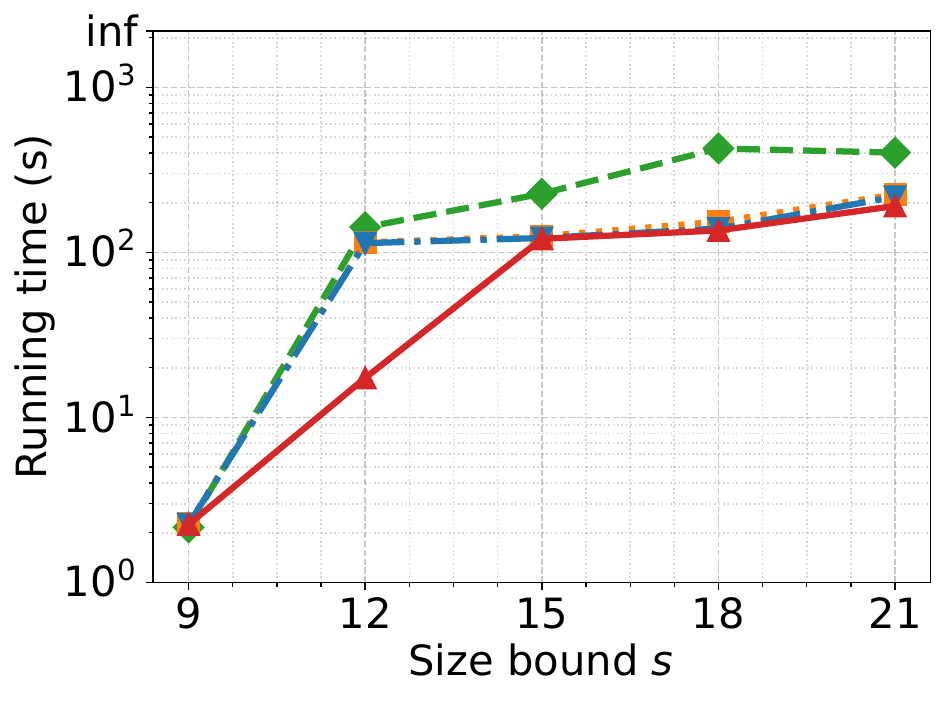}
    }

    \caption{Evaluate different pruning rules}
    \label{fig:DU}
\end{figure}

\noindent\textbf{Scalability}. In this experiment, we evaluate the scalability of SCSHEP$^+$ and SCSHEV$^+$ by randomly sampling between 20\% and 100\% of the edges from the original graph. Fig.~\ref{fig:Scalability} reports the results on the Amazon and Aminer datasets with $s=12$. As the graph size increases, the running time also increases accordingly.

\begin{figure}[h]
    \centering
    \vspace{-10pt}
    \subfigure{
        \includegraphics[width=.16\textwidth]{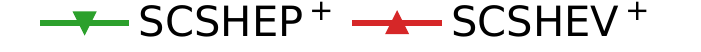}
    }
    \\[-10pt]
    \subfigure[Amazon]{
        \label{fig:AmazonDT}
        \includegraphics[width=.20\textwidth]{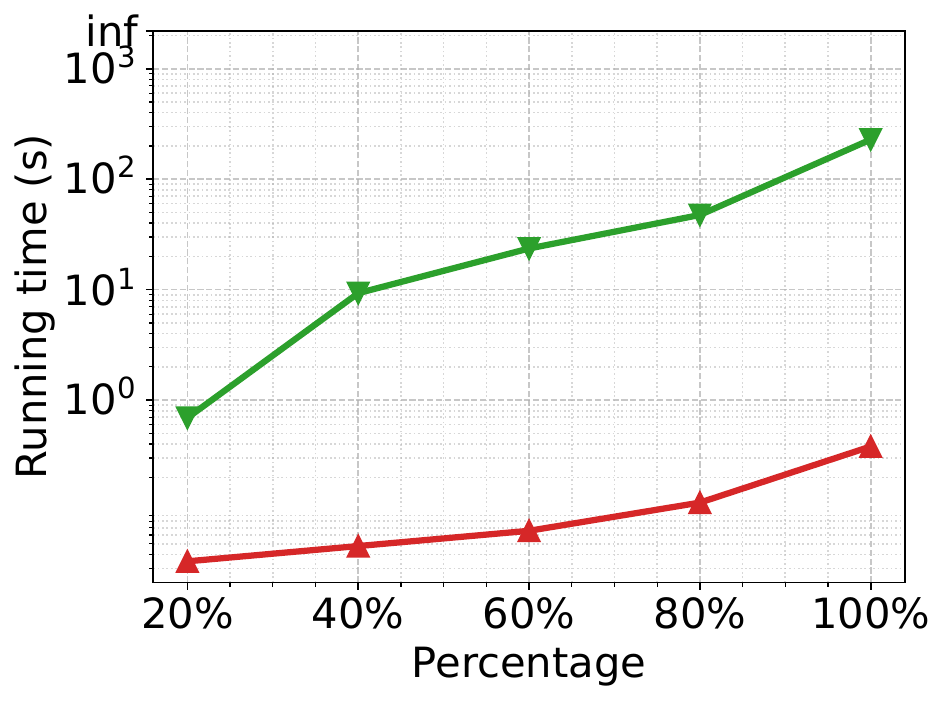}
    }
    \subfigure[Aminer]{
        \label{fig:AminerDT}
        \includegraphics[width=.20\textwidth]{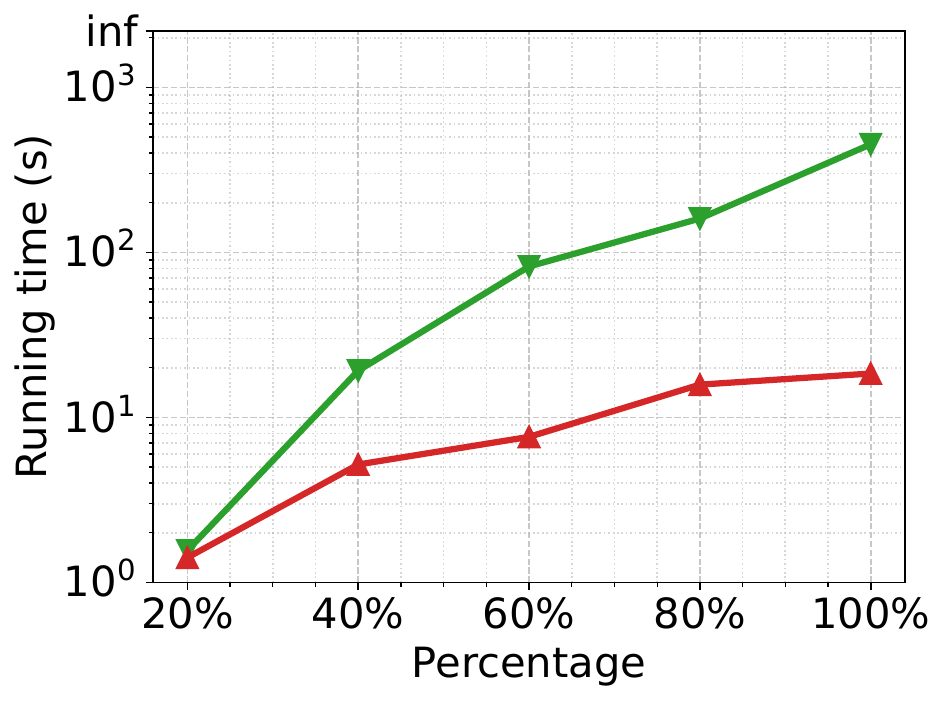}
    }

    \caption{Scalability}
    \label{fig:Scalability}
\end{figure}

\section{Related Works}
\label{sec7}
In this section, we review three categories of related works.

\noindent
\textbf{Size-bounded community search in homogeneous graphs.} In homogeneous graphs, many problems are related to community size, which can be broadly classified into the following categories. First, some studies impose explicit size constraints and aim to find the most cohesive communities satisfying conditions. 
For example, in the $k$-core model, the goal is to maximise $k$ under a given size bound. \cite{sozio2010community} is the first to introduce an upper bound on community size, proposing two greedy algorithms to find the near-optimal minimum-degree subgraph. 
Recently, \cite{yao2021efficient} extends this by introducing a size range $[l,h]$ and proposes both an improved greedy strategy and an exact algorithm. 
\cite{zhang2023size} follows a similar approach but adopts the $k$-truss model instead of the $k$-core. Second, some approaches fix cohesiveness and aim to find the smallest possible community satisfying it. \cite{li2020efficient} develops the PSA algorithm based on the $k$-core, while \cite{zhu2019pivotal} focuses on minimising $k$-truss-based communities. Finally, some studies~ \cite{liu2021efficient} jointly enforce both size and cohesiveness constraints (without necessarily maximising $k$ or minimising size). 

\noindent
\textbf{Community search in HINs.} \cite{fang2021cohesive} categorises heterogeneous graphs into two types: bipartite graphs and general HINs. In recent years, bipartite graphs have gained increasing attention, leading to the development of various cohesive subgraph models, including the ($\alpha$,$\beta$)-core \cite{liu2020efficient}, $k$-bitruss \cite{wang2022towards}, and biclique \cite{lyu2020maximum, chen2021efficient, yu2025finding}. Based on these models, numerous community search methods have been proposed \cite{wang2021efficient, wang2021discovering, abidi2022searching}. For general HINs, several models have been introduced to support community search, such as the ($k$, $\mP$)-core \cite{fang2020effective}, ($k$, $\mP$)-truss \cite{yao2021efficient}, motif-clique \cite{hu2019discovering, zhou2024efficient}, relational community \cite{jian2020effective, liu2025searching}, and densest subgraph \cite{chen2023densest}. In addition, some methods focus on extra-specific objectives like significance \cite{liu2024sach} and influence \cite{zhou2023influential}.

\noindent
\textbf{Size-bounded community search in HINs.} So far, only \cite{zhang2024size} has addressed size-bounded community search in bipartite graphs, which is a special case of HINs. This work is primarily based on the $(\alpha,\beta)$-core model. Since the model involves two cohesiveness parameters, it is difficult to compare the cohesiveness of different cores (e.g., a $(4,2)$-core vs. a $(3,3)$-core). As a result, \cite{zhang2024size} treats cohesiveness as a constraint rather than an optimisation objective.

\section{Conclusion}
\label{conclusion}
In this paper, we introduce the most cohesive size-bounded community search problem for HIN data, which is known to be NP-hard. To tackle this problem, we develop the \textsf{kcB\&B} framework, which naturally generates $s$ node sets with a complexity of $\mathcal{O}^*(\binom{n}{s})$. Based on this framework, two exact algorithms are proposed that enumerate edge sets and node sets, respectively. To enhance efficiency, we first develop a polynomial-time algorithm to find a feasible solution as a global lower bound by leveraging the properties of HINs. Additionally, based on triangles, we develop bounding and reduction techniques for both edge-based and node-based approaches, which, when combined with the lower bound, can significantly reduce the search space. In addition, \textsf{kcB\&B} enables early termination and branching strategies, further enhancing overall efficiency. Extensive experiments validate the effectiveness and efficiency of the proposed methods.

\clearpage

\bibliographystyle{ACM-Reference-Format}

\end{document}